%% file: paper.tex
\documentclass[runningheads]{llncs}

\usepackage{stmaryrd}
\usepackage{yhmath}
\usepackage{mathpartir}
\usepackage{textcomp}
\usepackage{amsmath, amssymb}
\usepackage{upgreek}
\usepackage{mathrsfs}

\usepackage{comment}

\includecomment{journal}
\excludecomment{conference} 

\usepackage{url}

\include{mymacros}

\usepackage{mathpartir}
\usepackage{orcidlink}
\usepackage{graphicx}
\usepackage{subcaption}

\usepackage{tikz}
\usetikzlibrary{arrows.meta, positioning, calc, shapes.geometric}

\def \config#1{\langle #1 \rangle}
\newcommand{\block}{\mathit{block}} 
\newcommand{\s}{\pi} 
\newcommand{\cons}{\!:\!}
\newcommand{\flow}{\mathit{flow}}

\newcommand{\entry}{\mathit{entry}}
\newcommand{\exit}{\mathit{exit}}

\newcommand{\context}{\mathsf{ctx}}
\newcommand{\cont}{\mathsf{cont}}
\newcommand{\nul}{[\:]}
\newcommand{\emptystate}{\epsilon}
\newcommand{\emptystore}{\epsilon}

\newcommand{\lh}{\leftharpoondown}
\newcommand{\rh}{\rightharpoonup}
\newcommand{\id}{\epsilon} 

\def\defemb#1#2{\expandafter\def\csname #1\endcsname
                              {\relax\ifmmode #2\else\hbox{$#2$}\fi}}
\defemb{cA}{{\cal A}}
\defemb{cB}{{\cal B}}
\defemb{cC}{{\cal C}}
\defemb{cD}{{\cal D}}
\defemb{cE}{{\cal E}}
\defemb{cF}{{\cal F}}
\defemb{cG}{{\cal G}}
\defemb{cH}{{\cal H}}
\defemb{cI}{{\cal I}}
\defemb{cJ}{{\cal J}}
\defemb{cL}{{\cal L}}
\defemb{cM}{{\cal M}}
\defemb{cN}{{\cal N}}
\defemb{cO}{{\cal O}}
\defemb{cP}{{\cal P}}
\defemb{cQ}{{\cal Q}}
\defemb{cR}{{\cal R}}
\defemb{cS}{{\cal S}}
\defemb{cT}{{\cal T}}
\defemb{cU}{{\cal U}}
\defemb{cV}{{\cal V}}
\defemb{cW}{{\cal W}}
\defemb{cX}{{\cal X}}
\defemb{cZ}{{\cal Z}}

\newcommand{\kw}[1]{\textcolor{blue}{\textbf{#1}}}

\title{A Reversible Semantics for Janus%
}
\titlerunning{A Reversible Semantics for Janus}

\author{Ivan Lanese\inst{1}\orcidlink{0000-0003-2527-9995} 
Germ\'an Vidal\inst{2}\orcidlink{0000-0002-1857-6951}}

\authorrunning{I.~Lanese and G.~Vidal}

\institute{
  Olas Team, University of Bologna \& Inria - Université Côte d’Azur, Bologna, Italy
  \and
  VRAIN, Universitat Polit\`ecnica de Val\`encia, Valencia, Spain
}

\begin{document}

\maketitle

\begin{abstract}
Janus is a paradigmatic example of a reversible programming language. Indeed, Janus programs can be executed backwards as well as forwards. However, its current small-step semantics (useful, e.g., for debugging or as a basis for extensions with concurrency primitives) is not reversible, since it loses information while computing forwards. E.g., it does not satisfy the Loop Lemma, stating that any reduction has an inverse, a main property of reversibility in process calculi, where a small-step semantics is commonly used. We present here a novel small-step semantics which is actually reversible, while remaining equivalent to the previous one. It involves the non-trivial challenge of defining a semantics based on a ``program counter'' for a high-level programming language.
\end{abstract}
	
\section{Introduction}

Reversible computation has emerged as a promising 
framework for addressing a number of limitations of current
computing paradigms \cite{ACGKKKLMMNPPPUV20}.
On the one hand, according to Landauer's principle
\cite{landauerIrreversibilityHeatGeneration1961},
the erasure of one bit of information in an irreversible 
computation necessarily dissipates a minimum amount of heat.
In contrast, Bennett \cite{revTM} showed that, theoretically,
a reversible computation---where no information is 
discarded---can be performed with less heat dissipation 
than an irreversible one. Thus, the development of 
reversible circuits could represent a revolution in 
reducing the energy consumption of large computing centers.

At the programming level, the Janus language 
\cite{lutz1986janus} represents one of the first 
attempts to formalize a \emph{reversible} programming language.
Janus performs (deterministic) computations in both 
forward and backward directions.
The development of reversible programming languages like 
Janus may provide an excellent opportunity to explore 
the relationship between logical reversibility, 
energy conservation, and computational expressiveness.
A formal semantics
for the language was first presented in 
\cite{YOKOYAMA201071,LoopSem,JanusSem}.
Unfortunately, the semantics of \cite{YOKOYAMA201071,LoopSem,JanusSem}
follows the so-called \emph{big-step} style, 
which is less appropriate for some applications, e.g.,
modeling stepwise execution, program debugging,
designing an abstract machine for the language,
modeling concurrency and interleaving (for
a future
concurrent extension of Janus), analyzing
non-termination, etc.

Recently, \cite{LLS24} has introduced a 
small-step semantics for Janus that is more
appropriate for the applications mentioned above.
This semantics employs \emph{contexts} to identify the 
next statement to be executed. In particular, the 
configurations include the code to be executed, 
which is transformed as execution progresses. 
These features, however, make the semantics 
irreversible. 
E.g., a Janus conditional statement of the form 
``$\ms{if}~ e_1~\ms{then}~\ms{skip}~\ms{else}~ s_2~\ms{fi}~e_2$'' 
is reduced to $\ms{skip}$ when the guard $e_2$ is true (in~\cite{LLS24} the $\ms{then}$ branch is first reduced to $\ms{skip}$ when the condition $e_1$ is true before evaluating the guard).
This step, however, is irreversible: we cannot obtain 
``$\ms{if}~ e_1~\ms{then}~\ms{skip}~\ms{else}~ s_2~\ms{fi}~e_2$''
from $\ms{skip}$.
Janus is thus reversible at the level of programs (or statements), in the sense that we can execute any program backwards, recovering input data from results, but not at the level of semantics, since single reduction steps cannot be inverted. This is in strike contrast with the semantics of calculi and languages for concurrency, where the Loop Lemma ensures that single steps have a reverse, e.g., in RCCS~\cite[Lemma 6]{rccs}, CCSK~\cite[Proposition 5.1]{revUlidowski}, or reversible Erlang~\cite[Lemma 11]{LaneseNPV18}.
This prevents applications 
such as reversible debugging, the definition of 
backtracking strategies, etc. Therefore, 
in this work, we tackle the definition of 
a \emph{reversible} small-step semantics for Janus.

First, we introduce a simple small-step semantics 
for Janus that uses a stack to model subcomputations,
and which is equivalent to previous semantics for
the language.
In order to make this semantics reversible, one
could add a ``history'' to the configurations, as it has 
been done to define a reversible semantics for the concurrent calculi and languages mentioned above.
%
This solution, however, is not very satisfactory, not 
only because of the cost in memory usage to store the history,
but also because it is not strictly necessary when the
considered language is reversible (the case of Janus).

Thus, in this paper, we explore an alternative idea.
Specifically, our proposal is based on replacing in the
configurations the component with the code to be executed 
with a ``program counter'' that points to the next 
statement to be executed. 
At first glance, one 
might think that a step could be ``undone'' simply 
by moving the program counter back to the previous position.
The actual definition of the
reversible semantics, however, is far from trivial.
In particular, given that Janus is a high-level language, 
in some cases it is not easy to determine how a program counter should 
be incremented. On the other hand, backward rules pose 
certain problems for determining the correct flow when 
there are several ways to go backwards from a given 
point in the execution, e.g., when undoing procedure calls.
%
Our small-step
semantics is made of a \emph{forward} transition relation, which 
models normal execution (i.e., which is a conservative extension
of the previous small-step semantics), and a \emph{backward} 
transition relation, which 
can be used to ``undo'' the steps of a (forward)
computation.

\begin{conference}
More details and missing proofs can be found in the extended version \cite{LV26arxiv}.
\end{conference}

\section{The Reversible Language Janus} \label{sec:janus}

In this section, we introduce the syntax and semantics of 
Janus~\cite{YOKOYAMA201071,LoopSem,JanusSem}, a reversible 
(imperative) language in which statements can be inverted 
and executed in both forward and backward directions.

\begin{figure}[t]
\begin{center}
$
\begin{array}{rcl@{~~~~~~}rcl@{ }}
\mi{p} & ::= & s ~(\ms{procedure} ~id~s)^{+} & \\
\mi{s} & ::= & x \oplus = e \, \mid \, x[e]\oplus=e \,
\mid \ms{if}~e ~\ms{then} ~s~\ms{else}~s~\ms{fi}~e \\
& \mid & \ms{from}~e~\ms{do}~s~\ms{loop}~s~\ms{until}~e \,
\mid \ms{call}~id \, \mid \, \ms{uncall}~id\mid \, \ms{skip} \, \mid \, s~s\\
\mi{e} & ::= & c \mid x  \mid x[e] \mid e \odot e \qquad \qquad \qquad \qquad \oplus ::=  + \mid - \mid \hat{~}\\
\mi{c} & ::= & - 2147483648\mid \dots \mid 0 \mid \dots \mid 2147483647 \\
\odot & ::= & \oplus \, \mid \, * \, \mid\, / \, \mid\,  \% \, \mid \, \& \, \mid\, \&\& \, \mid\,  \|  \, \mid\,  ''~|~'' \, 
\mid \, < \, \mid \, > \, \mid \, = \, \mid \, ! = \, \mid \, <= \, \mid\, >= &
\end{array}
$
\end{center}
\caption{Language syntax} \label{syntax}
\end{figure}

The syntax of Janus is 
in Fig.~\ref{syntax}. 
Janus is a simple imperative language 
with assignments, conditionals, loops, and procedure calls. 
A program consists of a statement (the main body) 
and a sequence of procedure declarations, 
introduced by the keyword $\ms{procedure}$ followed by the 
procedure name (an identifier). 

The reversible assignment instruction takes the form 
``$x \oplus = e$'' (or ``$x[e_1] \oplus = e_2$'' for array 
element assignments), which is equivalent to 
$x = x \oplus e$ (resp.\ $x[e_1] = x[e_1]\oplus e_2$), 
where $\oplus$ denotes an operator among plus, minus, and bitwise xor. 
Here, Janus requires that the left-hand side variable, 
$x$, does not occur in the expression $e$ 
(resp.\ in the expressions $e_1$ and $e_2$).

Reversible conditionals take the form 
``$\ms{if}~e_1~\ms{then}~s_1~\ms{else}~s_2~\ms{fi}~e_2$.'' 
If the test $e_1$ evaluates to $\mathit{true}$, then $s_1$ is executed, 
and the assertion $e_2$ must also evaluate to $true$. 
Similarly, if $e_1$ evaluates to $\mathit{false}$, then $s_2$
is executed, and the assertion $e_2$ must also be false.

Reversible loops have the form 
``$\ms{from}~e_1~\ms{do}~s_1~\ms{loop}~s_2~\ms{until}~e_2$.'' 
Initially, the assertion $e_1$ must evaluate to $\mathit{true}$. 
If so, $s_1$ is executed, followed by the evaluation of 
the test $e_2$. If $e_2$ is true, the loop terminates. 
Otherwise, $s_2$ is executed and control returns to $e_1$, 
which must evaluate to false in this case; 
execution then proceeds to execute $s_1$ and evaluate $e_2$ again,
and so forth (until $e_2$ is true).

Finally, we assume in this paper that procedures have neither
parameters nor local variables.  Consequently, a procedure call takes
the form ``$\ms{call}~\mathit{id}$'' and simply executes the procedure
body using the values of the program's global variables. Similarly,
``$\ms{uncall}~\mathit{id}$'' invokes the inverse procedure.

\begin{conference}
For a more detailed explanation of Janus, the reader is referred 
to \cite{JanusSem}.
\end{conference}

\begin{journal}


We now briefly describe the semantics of Janus. 
We primarily follow~\cite{JanusSem}, except for loops, 
where we adopt the simplified version 
of~\cite{YOKOYAMA201071,LoopSem}.


\begin{figure}[!t]
\begin{mathpar}
\inferrule*[left=Con] 
{ ~	} 
{ \tuple{\sigma}{c} \Downarrow  \llbracket c \rrbracket }     
\and
\inferrule*[left=Var] 
{ ~	} 
{ \tuple{\sigma}{ x} \Downarrow  \sigma(x)}      
\and
\inferrule*[left=Arr] 
{  \tuple{\sigma}{ e} \Downarrow	\tuple{\sigma}{ v} } 
{ \tuple{\sigma}{ x[e]} \Downarrow  x[v] }     
\and
\inferrule*[left=Bop] 
{ \tuple{\sigma}{ e_1} \Downarrow	 v_1 \\
\tuple{\sigma}{ e_2} \Downarrow	 v_2\\
\llbracket \odot \rrbracket(v_1,v_2) = v
} 
{ \tuple{\sigma}{ e_1 \odot e_2} \Downarrow  v }
\end{mathpar}
\caption{Semantics of expressions} \label{fig:exp}
\end{figure}

First, Figure~\ref{fig:exp} shows the rules for the evaluation of
expressions. A judgement of the form
$\tuple{\sigma}{e} \Downarrow  v$
denotes the evaluation of expression $e$ under \emph{store} 
$\sigma$, producing the value $v$. 
A \emph{store} $\sigma$ is a function from variable names and indexed variable names to values. We denote with $\sigma[x \mapsto v]$ the update of $\sigma$ which assigns value $v$ to variable $x$. 
For simplicity we assume that $\sigma$ is defined for all the variables used in the program, and we omit the bindings of variables whose value is zero.    
We let $\emptystore$ denote an empty store, i.e., an store that
maps every variable to zero.
Here, $\llbracket \_ \rrbracket$ denotes an \emph{interpretation} 
function, e.g., $\llbracket \mbox{\tt +} \rrbracket (2,3) =
2 + 3$, where ``$+$'' is the standard addition on integers.

\begin{figure}[!t]
\begin{mathpar}
\inferrule*[left=AssVar] 
{  \tuple{\sigma}{e} \Downarrow	v } 
{ \tuple{\sigma}{ x~\oplus=e} \Downarrow \sigma[x \mapsto \llbracket \oplus \rrbracket(\sigma(x),v)] }     
\and
\inferrule*[left=AssArr] 
{ \tuple{\sigma}{ e_l} \Downarrow	 v_l \\
\tuple{\sigma}{ e} \Downarrow	 v} 
{ \tuple{\sigma}{ x[e_l]~\oplus=e} \Downarrow \sigma[x[v_l] \mapsto \llbracket \oplus \rrbracket(\sigma(x[v_l]),v)] }     
\and
\inferrule*[left=Call] 
{ \tuple{\sigma}{ \Gamma(id)} \Downarrow  \sigma'  } 
{ \tuple{\sigma}{ \ms{call}~id} \Downarrow  \sigma'  }
\and
\inferrule*[left=UnCall] 
{  \tuple{\sigma}{\mathcal{I}\llbracket \Gamma(id) \rrbracket}  \Downarrow \sigma' } 
{ \tuple{\sigma}{ \ms{uncall}~id} \Downarrow  \sigma'  }
\and
\inferrule*[left=Seq] 
{
\tuple{\sigma}{ s_1}  \Downarrow  \sigma'\\
\tuple{\sigma'}{ s_2}  \Downarrow  \sigma''
} 
{ \tuple{\sigma}{ s_1~s_2} \Downarrow  \sigma''  }   
\and
\inferrule*[left=IfTrue] 
{ \tuple{\sigma}{ e_1} \Downarrow	 v_1 \\ \ms{is\_true?}(v_1) \\
\tuple{\sigma}{ s_1}\Downarrow	 \sigma' \\
\tuple{\sigma'}{ e_2} \Downarrow v_2 \\ \ms{is\_true?}(v_2)
} 
{ \tuple{\sigma}{ \ms{if}~ e_1 ~\ms{then}~ s_1 ~\ms{else}~s_2~\ms{fi} ~e_2} \Downarrow \sigma'  }     
\and
\inferrule*[left=IfFalse] 
{ \tuple{\sigma}{ e_1} \Downarrow	 v_1 \\ \ms{is\_false?}(v_1) \\
\tuple{\sigma}{ s_2}\Downarrow \sigma' \\
\tuple{\sigma'}{ e_2} \Downarrow v_2 \\ \ms{is\_false?}(v_2)
} 
{ \tuple{\sigma}{ \ms{if}~ e_1 ~\ms{then}~ s_1 ~\ms{else}~s_2~\ms{fi} ~e_2} \Downarrow  \sigma' }
\and
\inferrule*[left=LoopMain] 
{ 
\tuple{\sigma}{ e_1} \Downarrow	 v_1\\ 
\ms{is\_true?}(v_1) \\
\tuple{\sigma}{ s_1 } \Downarrow	 \sigma' \\
\tuple{\sigma'}{ (e_1,s_1,e_2,s_2)}\Downarrow	 \sigma''
} 
{ \tuple{\sigma}{ \ms{from}~ e_1 ~\ms{do}~ s_1 ~\ms{loop}~s_2~\ms{until}~e_2} \Downarrow  \sigma''  }
\and
\inferrule*[left=LoopBase] 
{ 
\tuple{\sigma}{ e_2} \Downarrow	 v_2\\ 
\ms{is\_true?}(v_2) 
} 
{ 
\tuple{\sigma}{ (e_1,s_1,e_2,s_2)} \Downarrow	 \sigma
}
\and
\inferrule*[left=LoopRec] 
{ 
\tuple{\sigma}{ e_2} \Downarrow	 v_2 \\
\ms{is\_false?}(v_2)\\ 
\tuple{\sigma}{ s_2}  \Downarrow	 \sigma'\\
\tuple{\sigma'}{ e_1} \Downarrow	 v_1 \\
\ms{is\_false?}(v_1)\\ 
\tuple{\sigma'}{ s_1}  \Downarrow	 \sigma''\\
\tuple{\sigma''}{ (e_1,s_1,e_2,s_2)}  \Downarrow	 \sigma'''
} 
{ 
\tuple{\sigma}{ (e_1,s_1,e_2,s_2)} \Downarrow	 \sigma'''  
}
\and
\inferrule*[left=Skip] 
{ ~ } 
{ \tuple{\sigma}{ \ms{skip}} \Downarrow \sigma }
\end{mathpar}
\caption{Big-step semantics of sequential Janus} \label{fig:bigstep}
\end{figure}

Fig.~\ref{fig:bigstep} shows the (big-step) 
semantics for statements. A judgement 
$\tuple{\sigma}{s} \Downarrow  \sigma'$ can be read as follows:
the execution of statement $s$ under store $\sigma$
produces a new program state denoted by store $\sigma'$.

Let us briefly describe the transition rules.
Assignments are dealt with rules \rules{AssVar} 
and \rules{AssArr}. Here, $(\oplus=)$ stands for
$(+=)$, $(-=)$, and $(\hat{~}=)$. 

Rule \rules{Call} handles procedure calls, $\ms{call}~id$, 
by executing the body $\Gamma(id)$. 
We assume $\Gamma$ maps procedure identifiers to 
their bodies. A procedure uncall (defined by rule 
\rules{UnCall}) proceeds analogously but executes
the \emph{inversion} of $\Gamma(id)$, which is
computed using the auxiliary function
$\mathcal{I}$ (defined in 
Fig.~\ref{fig:inverter}).\footnote{We note that, here, 
we consider rule \rules{Uncall} as defined in \cite{LLS24}, 
which is slightly different but equivalent to the original 
rule in cite{JanusSem}.}

%

The execution of sequences is dealt with in the obvious way by rule
\rules{Seq}.

The execution of conditionals is defined by rules
\rules{IfTrue} and \rules{IfFalse}, depending on whether the
initial test (and, thus, the final asserion) is true or false. 
Auxiliary predicates $\ms{is\_true?}(v)$ and $\ms{is\_false?}(v)$ 
check the truth value of $v$. 
	
The execution of a loop is defined by three rules: 
a rule for entering the loop, \rules{LoopMain}, a rule for 
exiting, \rules{LoopBase}, and a rule for iteration, 
\rules{LoopRec}. In order to enter the loop, rule \rules{LoopMain}
requires the initial assertion $e_1$ to be true. 
Then, $s_1$ is executed. Now, if the final test $e_2$ is true,
the loop ends (rule \rules{LoopBase}). Otherwise, rule
\rules{LoopRec} applies which executes $s_2$ and $s_1$,
requiring both $e_2$ and $e_1$ to be false. 


Finally, rule \rules{Skip} does not change the current store.

For a deeper discussion, the interested reader is referred to \cite{YOKOYAMA201071,LoopSem,JanusSem}.
	
\begin{figure}[t]
\begin{center}
$
\begin{array}{rcl@{~~~~~~}rcl@{ }}
\mathcal{I}_{op}\llbracket + \rrbracket ::= - \hspace{1,5cm} 
\mathcal{I}_{op}\llbracket - \rrbracket & ::= & + \hspace{1,5cm}
\mathcal{I}_{op}\llbracket \hat{~} \rrbracket ::= \hat{~} \\
\mathcal{I}\llbracket x\oplus= e \rrbracket & ::= & x~\mathcal{I}_{op}\llbracket \oplus \rrbracket= e \\
\mathcal{I}\llbracket  x[e_1]\oplus=e_2 \rrbracket & ::= & x[e_1]~\mathcal{I}_{op}\llbracket \oplus \rrbracket= e_2 \\
\mathcal{I}\llbracket \ms{if}~e_1 ~\ms{then} ~s_1~\ms{else}~s_2~\ms{fi}~e_2 \rrbracket & ::= & \ms{if}~e_2 ~\ms{then} ~\mathcal{I}\llbracket s_1 \rrbracket~\ms{else}~\mathcal{I}\llbracket s_2 \rrbracket~\ms{fi}~e_1 \\
\mathcal{I}\llbracket \ms{from}~e_1~\ms{do}~s_1~\ms{loop}~s_2~\ms{until}~e_2 \rrbracket & ::= & \ms{from}~e_2~\ms{do}~\mathcal{I}\llbracket s_1 \rrbracket~\ms{loop}~\mathcal{I}\llbracket s_2 \rrbracket~\ms{until}~e_1 \\		
\mathcal{I}\llbracket \ms{call}~id \rrbracket & ::= & \ms{uncall}~id \\
\mathcal{I}\llbracket \ms{uncall}~id \rrbracket & ::= &   \ms{call}~id \\
\mathcal{I}\llbracket \ms{skip} \rrbracket & ::= & \ms{skip} \\
\mathcal{I}\llbracket s_1~s_2 \rrbracket & ::= & \mathcal{I}\llbracket s_2 \rrbracket ~ \mathcal{I}\llbracket s_1 \rrbracket \\
\mathcal{I}\llbracket \ms{procedure}~id~s \rrbracket & ::= & \ms{procedure} ~id^{-1}~\mathcal{I}\llbracket s \rrbracket \\
\end{array}
$
\end{center}
\caption{Inverter function for Janus statements} \label{fig:inverter}
\end{figure}
					
\section{A Small-Step Semantics for Janus} \label{sec:smallstep}

\end{journal}


\begin{journal}
Now, we introduce a \emph{small-step} semantics for Janus.
\end{journal}
\begin{conference}
Now, we introduce a \emph{small-step} semantics for Janus
(in contrast to the big-step 
semantics introduced in 
\cite{YOKOYAMA201071,LoopSem,JanusSem}).
\end{conference}
Moving from a big-step semantics to a small-step one 
typically involves adding a \emph{stack} to the configurations 
(see, e.g., \cite{Sestoft97,SPJ97}). 
Alternatively, one can define a small-step 
semantics using \emph{contexts} to specify where reductions can occur,
as done in \cite{LLS24}.
In this work, though, we prefer stacks, as they simplify 
the definition of a reversible version. Furthermore, they may help in clarifying the reason leading to an execution error,
which could be useful 
in a potential extension of the semantics involving 
error handling.

\begin{conference}
\begin{figure}[!t]
\begin{mathpar}
\inferrule*[left=Con] 
{ ~	} 
{ \tuple{\sigma}{c} \Downarrow  \llbracket c \rrbracket }     
\and
\inferrule*[left=Var] 
{ ~	} 
{ \tuple{\sigma}{ x} \Downarrow  \sigma(x)}      
\and
\inferrule*[left=Arr] 
{  \tuple{\sigma}{ e} \Downarrow	\tuple{\sigma}{ v} } 
{ \tuple{\sigma}{ x[e]} \Downarrow  x[v] }     
\and
\inferrule*[left=Bop] 
{ \tuple{\sigma}{ e_1} \Downarrow	 v_1 \\
\tuple{\sigma}{ e_2} \Downarrow	 v_2\\
\llbracket \odot \rrbracket(v_1,v_2) = v
} 
{ \tuple{\sigma}{ e_1 \odot e_2} \Downarrow  v }
\end{mathpar}
\caption{Semantics of expressions} \label{fig:exp}
\end{figure}
\end{conference}

For the evaluation of expressions, we consider the same rules 
of the original big-step semantics \cite{JanusSem},
which are shown in Fig.~\ref{fig:exp}.\footnote{Note that this 
means that the evaluation of expressions
is not reversible (but considered a one-step
evaluation).}
\begin{conference}
A judgement of the form
$\tuple{\sigma}{e} \Downarrow  v$
denotes the evaluation of expression $e$ under \emph{store} 
$\sigma$, producing the value $v$. 
A \emph{store} $\sigma$ is a function from variable names and indexed variable names to values. We denote with $\sigma[x \mapsto v]$ the update of $\sigma$ which assigns value $v$ to variable $x$. 
For simplicity we assume that $\sigma$ is defined for all the variables used in the program, and we omit the bindings of variables whose value is zero.    
We let $\emptystore$ denote an empty store, i.e., an store that
maps every variable to zero.
Here, $\llbracket \_ \rrbracket$ denotes an \emph{interpretation} 
function, e.g., $\llbracket \mbox{\tt +} \rrbracket (2,3) =
2 + 3$, where ``$+$'' is the standard addition on integers.
\end{conference}


A \emph{configuration} in our semantics has the form
$\config{\sigma,s,\s}$, where $\sigma$ is a global
store, $s$ is a statement (the \emph{control} 
of the configuration), and $\s$ is a \emph{stack}, 
i.e., a list of expressions of the form 
$\ms{seq}(e_2)$,
$\ms{if\_true}(e_2)$, 
$\ms{if\_false}(e_2)$,
$\ms{loop1}(e_1,s_1,e_2,s_2)$, and 
$\ms{loop2}(e_1,s_1,e_2,s_2)$.
%
%
Given a Janus program $p = (s~p_1~\ldots~p_n)$, 
an \emph{initial} configuration has the form
$\config{\emptystate,s,\nul}$, where $\nul$ is 
an empty list representing the initial, empty stack.

\begin{figure}[t]
\begin{mathpar}
\inferrule*[left=AssVarS] 
{  \tuple{\sigma}{e} \Downarrow	 v } 
{ \config{\sigma,x~\oplus=e,\s} \rightarrow \config{\sigma[x \mapsto \llbracket \oplus \rrbracket(\sigma(x),v)],\ms{skip},\s} }     
\and
\inferrule*[left=AssArrS] 
{ 
\tuple{\sigma}{ e_l} \Downarrow	 v_l \\
\tuple{\sigma}{ e} \Downarrow	 v
} 
{
\config{\sigma,x[e_l]~\oplus=e,\s} \rightarrow 
\config{\sigma[x[v_l] \mapsto \llbracket \oplus \rrbracket(\sigma(x[v_l]),v)],\ms{skip},\s} 
}     
\and
\inferrule*[left=CallS] 
{  ~} 
{ \config{\sigma,\ms{call}~id,\s} \rightarrow  \config{\sigma,\Gamma(id),\s}  }
\and
\inferrule*[left=UnCallS] 
{   ~ } 
{ \config{\sigma,\ms{uncall}~id,\s} \rightarrow 
\config{\sigma,\mathcal{I}\llbracket \Gamma(id) \rrbracket,\s}
}
\and
\inferrule*[left=Seq1]
{ ~ }
{
\config{\sigma,s_1~s_2,\s} \rightarrow  
\config{\sigma,s_1,\ms{seq}(s_2)\cons\s}
}
\and
\inferrule*[left=Seq2]
{ ~ }
{
\config{\sigma,\ms{skip},\ms{seq}(s_2)\cons\s} \rightarrow  \config{\sigma,s_2,\s} 
}
\end{mathpar}
\caption{Small-step semantics: rules for basic constructs} \label{fig:smallstep}
\end{figure}

\begin{conference}
\begin{figure}[t]
\begin{center}
$
\begin{array}{rcl@{~~~~~~}rcl@{ }}
\mathcal{I}_{op}\llbracket + \rrbracket ::= - \hspace{1,5cm} 
\mathcal{I}_{op}\llbracket - \rrbracket & ::= & + \hspace{1,5cm}
\mathcal{I}_{op}\llbracket \hat{~} \rrbracket ::= \hat{~} \\
\mathcal{I}\llbracket x\oplus= e \rrbracket & ::= & x~\mathcal{I}_{op}\llbracket \oplus \rrbracket= e \\
\mathcal{I}\llbracket  x[e_1]\oplus=e_2 \rrbracket & ::= & x[e_1]~\mathcal{I}_{op}\llbracket \oplus \rrbracket= e_2 \\
\mathcal{I}\llbracket \ms{if}~e_1 ~\ms{then} ~s_1~\ms{else}~s_2~\ms{fi}~e_2 \rrbracket & ::= & \ms{if}~e_2 ~\ms{then} ~\mathcal{I}\llbracket s_1 \rrbracket~\ms{else}~\mathcal{I}\llbracket s_2 \rrbracket~\ms{fi}~e_1 \\
\mathcal{I}\llbracket \ms{from}~e_1~\ms{do}~s_1~\ms{loop}~s_2~\ms{until}~e_2 \rrbracket & ::= & \ms{from}~e_2~\ms{do}~\mathcal{I}\llbracket s_1 \rrbracket~\ms{loop}~\mathcal{I}\llbracket s_2 \rrbracket~\ms{until}~e_1 \\		
\mathcal{I}\llbracket \ms{call}~id \rrbracket & ::= & \ms{uncall}~id \\
\mathcal{I}\llbracket \ms{uncall}~id \rrbracket & ::= &   \ms{call}~id \\
\mathcal{I}\llbracket \ms{skip} \rrbracket & ::= & \ms{skip} \\
\mathcal{I}\llbracket s_1~s_2 \rrbracket & ::= & \mathcal{I}\llbracket s_2 \rrbracket ~ \mathcal{I}\llbracket s_1 \rrbracket \\
\mathcal{I}\llbracket \ms{procedure}~id~s \rrbracket & ::= & \ms{procedure} ~id^{-1}~\mathcal{I}\llbracket s \rrbracket \\
\end{array}
$
\end{center}
\caption{Inverter function for Janus statements \cite{JanusSem}} \label{fig:inverter}
\end{figure}
\end{conference}

The small-step semantics is defined using a set of 
transition rules. We split them in 
two groups: basic statements (Fig.~\ref{fig:smallstep})
and conditional and
loop statements 
(Fig.~\ref{fig:ifloop}).
A \emph{transition step} with the small-step semantics
has the form 
$\config{\sigma,s,\s} \to \config{\sigma',s',\s'}$.
Let us briefly explain the first set of rules (Fig.~\ref{fig:smallstep}):
\begin{itemize}
\item Rules \rules{AssVarS} and \rules{AssArrS} 
define the assignment. 
\rules{AssVarS} evaluates the new value of the variable and 
updates the store, while \rules{AssArrS} evaluates the index 
of the array and 
the new value and updates the value of the array 
in the store.

\item A procedure call (defined by rule \rules{CallS}) 
reduces to the procedure body, denoted by $\Gamma(id)$.
%
Similarly, a procedure uncall (defined by rule 
\rules{UnCallS}) proceeds in much
the same way, except that it reduces to the 
\emph{inversion} 
$\mathcal{I}\llbracket\Gamma(id)\rrbracket$ 
of the body of procedure $id$, 
retrieved using the mapping $\Gamma$. 
\begin{conference}
Here $\mathcal{I}$ is the Janus 
inverter function defined in Fig.~\ref{fig:inverter},
which given a statement computes another statement 
executing the inverse computation.
\end{conference}

\item Finally, the execution of a sequence is defined 
by rules \rules{Seq1} and \rules{Seq2}. Here, we assume 
that the sequence operator is right-associative 
so that $s_1~s_2~\ldots~s_n$ is
interpreted as $s_1~(s_2~(\ldots (s_{n-1}~s_n)\ldots))$.
Given a sequence $s_1~s_2$, rule \rules{Seq1} moves
$s_2$ to the stack (adding an element of the form
$\mathsf{seq}(s_2)$ on top of the stack) 
and starts the evaluation of $s_1$.
Once $s_1$ is fully evaluated, rule \rules{Seq2} moves
$s_2$ back to the control of the configuration.
\end{itemize}
\begin{figure}[t]
\begin{mathpar}
\inferrule*[left=IfTrue1] 
{ 
\tuple{\sigma}{e_1} \Downarrow	 v_1 \\ 
\ms{is\_true?}(v_1) \\
} 
{ \config{\sigma,\ms{if}~ e_1 ~\ms{then}~ s_1 ~\ms{else}~s_2~\ms{fi} ~e_2,\s} \rightarrow  \config{\sigma,s_1,\ms{if\_true}(e_2)\cons\s}  
}
\and 
\inferrule*[left=IfTrue2] 
{ 
\tuple{\sigma}{e_2} \Downarrow	 v_2 \\ 
\ms{is\_true?}(v_2) \\
}
{ \config{\sigma,\ms{skip},\ms{if\_true}(e_2)\cons\s}  
\rightarrow 
\config{\sigma,\ms{skip},\s}
}
\and 
\and 
\inferrule*[left=IfFalse1] 
{ 
\tuple{\sigma}{ e_1} \Downarrow	 v_1 \\ 
\ms{is\_false?}(v_1) \\
} 
{ \config{\sigma,\ms{if}~ e_1 ~\ms{then}~ s_1 ~\ms{else}~s_2~\ms{fi} ~e_2,\s} \rightarrow  \config{\sigma,s_2,\ms{if\_false}(e_2)\cons\s}  
}
\and 
\inferrule*[left=IfFalse2] 
{ 
\tuple{\sigma}{e_2} \Downarrow	 v_2 \\ 
\ms{is\_false?}(v_2) \\
}
{ \config{\sigma,\ms{skip},\ms{if\_false}(e_2)\cons\s}  
\rightarrow 
\config{\sigma,\ms{skip},\s}
}
\and
\inferrule*[left=LoopMainS] 
{ 
\tuple{\sigma}{e_1} \Downarrow	 v_1\\ 
\ms{is\_true?}(v_1) 
} 
{
\config{\sigma,\ms{from}~ e_1 ~\ms{do}~ s_1 ~\ms{loop}~s_2~\ms{until} ~e_2,\s} 
\rightarrow  
\config{\sigma,s_1,\ms{loop1}(e_1,s_1,e_2,s_2)\cons\s}  
}
\and
\inferrule*[left=LoopBaseS] 
{ 
\tuple{\sigma}{e_2} \Downarrow	 v_2\\ 
\ms{is\_true?}(v_2) 
} 
{
\config{\sigma,\ms{skip},\ms{loop1}(e_1,s_1,e_2,s_2)\cons\s}  
\rightarrow
\config{\sigma,\ms{skip},\s}  
}
\and
\inferrule*[left=Loop1] 
{ 
\tuple{\sigma}{e_2} \Downarrow	 v_2\\ 
\ms{is\_false?}(v_2) 
} 
{
\config{\sigma,\ms{skip},\ms{loop1}(e_1,s_1,e_2,s_2)\cons\s}  
\rightarrow
\config{\sigma,s_2,\ms{loop2}(e_1,s_1,e_2,s_2)\cons\s}  
}
\and
\inferrule*[left=Loop2] 
{ 
\tuple{\sigma}{e_1} \Downarrow	 v_1\\ 
\ms{is\_false?}(v_1) 
} 
{
\config{\sigma,\ms{skip},\ms{loop2}(e_1,s_1,e_2,s_2)\cons\s}  
\rightarrow
\config{\sigma,s_1,\ms{loop1}(e_1,s_1,e_2,s_2)\cons\s}  
}
\end{mathpar}
\caption{Small-step semantics: rules for if and for loop} \label{fig:ifloop}
\end{figure}
The transition rules for conditionals and for loops are shown in Fig.~\ref{fig:ifloop}. We distinguish the following
cases:
\begin{itemize}
\item For conditionals, first, either rule \rules{IfTrue1} or
\rules{IfFalse1} is applied, depending on whether the
condition $e_1$ evaluates to $\mathit{true}$ or 
$\mathit{false}$, respectively. Rule \rules{IfTrue1}
(resp.\ \rules{IfFalse1})
pushes a new element of the form 
$\mathsf{if\_true}(e_2)$ 
(resp.\ $\mathsf{if\_false}(e_2)$)
onto the stack.

\item Once the corresponding
statement is fully evaluated, rule \rules{IfTrue2} (resp.\
\rules{IfFalse2}) applies if the assertion evaluates to the
same value of the test. 

\item Regarding loops, rule \rules{LoopMain} enters the loop,
requiring the assertion $e_1$ to be true and starting
the execution of statement $s_1$. Furthermore, an
element of the form $\mathsf{loop1}(e_1,s_1,e_2,s_2)$ is
added on top of the stack.

\item Each iteration of the loop is formalized using
rules \rules{Loop1} and \rules{Loop2}. The first rule,
\rules{Loop1},
applies when $s_1$ is fully evaluated and the condition
$e_2$ evaluates to $\mathit{false}$. Here, $s_2$ is moved to the
control of the configuration and 
$\mathsf{loop1}(e_1,s_1,e_2,s_2)$ is replaced by
$\mathsf{loop2}(e_1,s_1,e_2,s_2)$ in the stack.
Rule \rules{Loop2} is the counterpart of rule \rules{Loop1}
but applies when $s_2$ is fully evaluated. In this case,
$e_1$ must be evaluated to $\mathit{false}$. Then,
$s_1$ is moved to the
control of the configuration and 
$\mathsf{loop2}(e_1,s_1,e_2,s_2)$ is changed back to
$\mathsf{loop1}(e_1,s_1,e_2,s_2)$ in the stack.

\item Finally, rule \rules{LoopBaseS} applies when $s_1$
is fully evaluated and, moreover, the test $e_2$ is 
true, terminating loop execution.
\end{itemize}
We denote with $\rightarrow^*$ the reflexive 
and transitive closure of $\rightarrow$.

\begin{conference}
The equivalence of the small-step semantics with 
the big-step semantics of \cite{YOKOYAMA201071,LoopSem,JanusSem}
(and, thus, the equivalence to the small-step semantics
based on contexts of \cite{LLS24}) can be found in the
extended version of this paper \cite{LV26arxiv}.
\end{conference}

\begin{journal}

Now, we prove the equivalence between the big-step and the 
small-step semantics, when the program terminates without 
errors.   
For this purpose, we follow a proof scheme similar to that
of \cite{Sestoft97}. 

In the following, we say that $s$ is a \emph{proper} statement
if $s\neq(e_1,s_1,e_2,s_2)$. Our first lemma shows that the
small-step semantics is \emph{complete}, i.e., it can simulate 
all proof derivations with the natural semantics.

\begin{lemma} \label{lemma:completeness}
  Let $\tuple{\sigma}{s} \Downarrow \sigma'$ be a proof with the
  big-step semantics of Figure~\ref{fig:bigstep}.
  Then, we have 
  \begin{itemize}
  \item[\rm (a)] $\config{\sigma,s,\s} \to^\ast \config{\sigma',\ms{skip},\s}$ 
  if $s$ is a proper statement, and
  \item[\rm (b)] $\config{\sigma,\ms{skip},\ms{loop1}(e_1,s_1,e_2,s_2)\cons\s} \to^\ast
  \config{\sigma',\ms{skip},\s}$ if $s = (e_1,s_1,e_2,s_2)$,
  \end{itemize} 
  where $\s$ is an arbitrary stack.
\end{lemma}

\begin{proof}
  We prove the claim by induction on the structure of the
  proof $\tuple{\sigma}{s} \Downarrow \sigma'$. For simplicity, 
  we implicitly assume that a sentence is \emph{proper} unless 
  otherwise stated. In particular, only rule \textsc{LoopMain}
  can introduce a non-proper statement, while only rules
  \textsc{LoopBase} and \textsc{LoopRec} can be applied to
  evaluate a non-proper statement. 
  We distinguish the following cases depending 
  on the applied rule:
  \begin{description}
  \item[\rm (\textsc{AssVar})] Then, we have 
  $\tuple{\sigma}{x~\oplus=e} \Downarrow 
  \sigma[x \mapsto \llbracket \oplus \rrbracket(\sigma(x),v)]$ 
  by rule \textsc{AssVar}. Hence, 
  $\config{\sigma,x~\oplus=e,\s} \to 
  \config{\sigma[x \mapsto \llbracket \oplus 
  \rrbracket(\sigma(x),v)],\ms{skip},\s}$ 
  by rule \textsc{AssVarS}.

  \item[\rm (\textsc{AssArr})] This case is perfectly analogous
  to the previous one by applying rule \textsc{AssArrS} instead.
  
  \item[\rm (\textsc{Call})] We have 
  $\tuple{\sigma}{\ms{call}~id} \Downarrow  \sigma'$,
  with $\tuple{\sigma}{\Gamma(id)} \Downarrow  \sigma'$.
  Then, the following derivation holds:
  \[
  \begin{array}{lll}
      & \config{\sigma,\ms{call}~id,\s} \\
  \to & \config{\sigma,\Gamma(id),\s} & \mbox{(by rule \textsc{CallS})}\\
  \to^\ast & \config{\sigma',\ms{skip},\s} & \mbox{(by ind.\ hypothesis)}
  \end{array}
  \]
  
  \item[\rm (\textsc{UnCall})] This case is perfectly analogous
  to the previous one by applying rule \textsc{UnCallS} instead.
  
  \item[\rm (\textsc{Seq})] We have 
  $\tuple{\sigma}{s_1~s_2} \Downarrow  \sigma''$, with 
  $\tuple{\sigma}{s_1}  \Downarrow  \sigma'$ and
  $\tuple{\sigma'}{s_2}  \Downarrow  \sigma''$. 
  Then, the following derivation holds:
  \[
  \begin{array}{lll}
      & \config{\sigma,s_1~s_2,\s} \\
  \to & \config{\sigma,s_1,\ms{seq}(s_2)\cons\s} & \mbox{(by rule \textsc{Seq1})}\\
\to^\ast & \config{\sigma',\ms{skip},\ms{seq}(s_2)\cons\s} & \mbox{(by ind.\ hypothesis on the first premise)}\\
  \to & \config{\sigma',s_2,\s} & \mbox{(by rule \textsc{Seq2})}\\
\to^\ast & \config{\sigma'',\ms{skip},\s} & \mbox{(by ind.\ hypothesis on the second premise)}\\  
  \end{array}
  \]

  \item[\rm (\textsc{IfTrue})] We have 
  $\tuple{\sigma}{ \ms{if}~ e_1 ~\ms{then}~ s_1 
  ~\ms{else}~s_2~\ms{fi} ~e_2} \Downarrow \sigma'$ with
  $\tuple{\sigma}{ e_1} \Downarrow	 v_1$, $\ms{is\_true?}(v_1)$,
  $\tuple{\sigma}{ s_1}\Downarrow \sigma'$,
  $\tuple{\sigma'}{ e_2} \Downarrow v_2$, and $\ms{is\_true?}(v_2)$.
Then, the following derivation holds:
  \[
  \begin{array}{lll}
      & \config{\sigma,\ms{if}~ e_1 ~\ms{then}~ s_1 
  ~\ms{else}~s_2~\ms{fi} ~e_2,\s} \\
  \to & \config{\sigma,s_1,\ms{if\_true}(e_2)\cons\s} & \mbox{(by rule \textsc{IfTrue1} since $\tuple{\sigma}{ e_1} \Downarrow v_1$}\\ 
  &&\mbox{~and $\ms{is\_true?}(v_1)$)}\\
  \to^\ast & \config{\sigma',\ms{skip},\ms{if\_true}(e_2)\cons\s} & \mbox{(by ind.\ hypothesis on $\tuple{\sigma}{ s_1}\Downarrow \sigma'$)}\\
  \to & \config{\sigma',s_2,\s} & \mbox{(by rule \textsc{IfTrue2}
  since $\tuple{\sigma'}{ e_2} \Downarrow v_2$}\\
  &&\mbox{~and $\ms{is\_true?}(v_2)$)}\\
  \end{array}
  \]
  
  \item[\rm (\textsc{IfFalse})] This case is perfectly analogous
  to the previous one by applying rules \textsc{IfFalse1} and
  \textsc{IfFalse2} from the small-step semantics.
  
  \item[\rm (\textsc{LoopMain})] We have 
  $\tuple{\sigma}{ \ms{from}~ e_1 ~\ms{do}~ s_1 ~\ms{loop}~s_2~\ms{until}~e_2} \Downarrow  \sigma''$
  with $\tuple{\sigma}{ e_1} \Downarrow v_1$,
  $\ms{is\_true?}(v_1)$, 
  $\tuple{\sigma}{ s_1 } \Downarrow	 \sigma'$, and
  $\tuple{\sigma'}{ (e_1,s_1,e_2,s_2)}\Downarrow \sigma''$.
  Then, the following derivation holds:
  \[
  \begin{array}{lll}
   & \config{\sigma,\ms{from}~ e_1 ~\ms{do}~ s_1 ~\ms{loop}~s_2~\ms{until}~e_2,\s} \\
  \to & \config{\sigma,s_1,\ms{loop1}(e_2,s_1,s_2,e_2)\cons\s} & \mbox{(by rule \textsc{LoopMainS} since }\\
  &&\mbox{~$\tuple{\sigma}{ e_1} \Downarrow v_1$ and $\ms{is\_true?}(v_1)$)}\\
  \to^\ast & \config{\sigma',\ms{skip},\ms{loop1}(e_2,s_1,s_2,e_2)\cons\s} & \mbox{(by ind.\ hypothesis on $\tuple{\sigma}{ s_1 } \Downarrow \sigma'$)}\\
  \to^\ast & \config{\sigma'',\ms{skip},\s} & \mbox{(by ind.\ hypothesis }\\
  &&\mbox{~on $\tuple{\sigma'}{ (e_1,s_1,e_2,s_2)}\Downarrow \sigma''$)}\\
  \end{array}
  \]

  \item[\rm (\textsc{LoopBase})] We have 
  $\tuple{\sigma}{ (e_1,s_1,e_2,s_2)} \Downarrow \sigma$,
  with $\tuple{\sigma}{e_2} \Downarrow v_2$ and 
  $\ms{is\_true?}(v_2)$. Then, the derivation
  $\config{\sigma,\ms{skip},\ms{loop1}(e_1,s_1,e_2,s_2)\cons\s} \to \config{\sigma,\ms{skip},\s}$ holds by applying
  rule \textsc{LoopBaseS}.

  \item[\rm (\textsc{LoopRec})] We have 
  $\tuple{\sigma}{(e_1,s_1,e_2,s_2)} \Downarrow \sigma'''$, with
  $\tuple{\sigma}{e_2} \Downarrow v_2$,
  $\ms{is\_false?}(v_2)$, 
  $\tuple{\sigma}{s_2} \Downarrow \sigma'$,
  $\tuple{\sigma'}{e_1} \Downarrow v_1$,
  $\ms{is\_false?}(v_1)$,  
  $\tuple{\sigma'}{s_1}  \Downarrow	 \sigma''$, and
  $\tuple{\sigma''}{ (e_1,s_1,e_2,s_2)} \Downarrow	 \sigma'''$.
  Then, the following derivation holds:
  \[
  \begin{array}{lll}
   & \config{\sigma,\ms{skip},\ms{loop1}(e_1,s_1,e_2,s_2)\cons\s}\\
   \to & \config{\sigma,s_2,\ms{loop2}(e_1,s_1,e_2,s_2)\cons\s}
   & \mbox{(by rule \textsc{Loop1} since $\tuple{\sigma}{e_2} \Downarrow v_2$ }\\
   &&\mbox{~and $\ms{is\_false?}(v_2)$)}\\
   \to^\ast & \config{\sigma',\ms{skip},\ms{loop2}(e_1,s_1,e_2,s_2)\cons\s}
   & \mbox{(by ind.\ hypothesis on $\tuple{\sigma}{s_2} \Downarrow \sigma'$)}\\
  \to & \config{\sigma',s_1,\ms{loop1}(e_1,s_1,e_2,s_2)\cons\s}
   & \mbox{(by rule \textsc{Loop2} since $\tuple{\sigma}{e_1} \Downarrow v_1$ }\\
   &&\mbox{~and $\ms{is\_false?}(v_1)$)}\\
  \to^\ast & \config{\sigma'',\ms{skip},\ms{loop1}(e_1,s_1,e_2,s_2)\cons\s}
   & \mbox{(by ind.\ hypothesis on $\tuple{\sigma'}{s_1} \Downarrow \sigma''$)}\\
   \to^\ast & \config{\sigma''',\ms{skip},\s} &
   \mbox{(by ind. hypothesis }\\
   &&\mbox{~on $\tuple{\sigma''}{ (e_1,s_1,e_2,s_2)} \Downarrow	 \sigma'''$)}
   \end{array}
   \]
  
  \item[\rm (\textsc{Skip})] We have $s = \ms{skip}$ and
  $\tuple{\sigma}{ \ms{skip}} \Downarrow \sigma$. Thus,
  $\config{\sigma,\ms{skip},\s} \to^\ast
  \config{\sigma,\ms{skip},\s}$ since $\to^\ast$ is reflexive. \qed
  \end{description}
\end{proof}
Let us now consider the soundness of the small-step semantics.
For this purpose, we first introduce the notion of a
\emph{balanced} derivation:

\begin{definition}[balanced derivation]
  A derivation $\config{\sigma,s,\s} \to^\ast 
  \config{\sigma',s',\s'}$ is balanced if $\s = \s'$ and every
  intermediate stack is equal or extends the initial one.
\end{definition}
In particular, any successful derivation of the form
$\config{\id,s,\nul} \to^\ast \config{\sigma,\ms{skip},\nul}$
is trivially balanced.
We also consider another type of derivations:

\begin{definition}[loop derivation]
  A derivation 
  $\config{\sigma,\ms{skip},\ms{loop1}(e_1,s_1,e_2,s_2)\cons\s}$
  $\to^\ast \config{\sigma',\ms{skip},\s}$ is called a
  \emph{loop derivation} if every intermediate stack
  is equal or extends either  
  $\ms{loop1}(e_1,s_1,e_2,s_2)\cons\s$
  or $\ms{loop2}(e_1,s_1,e_2,s_2)\cons\s$.
\end{definition}
Intuitively speaking, a loop derivation corresponds to
sequences of steps with rules \textsc{Loop1} and \textsc{Loop2},
ending with a final application of rule \textsc{LoopBaseS}.

\begin{definition}[trace]
  The sequence of transition rules applied in a derivation
  is called its \emph{trace}. A balanced (resp.\ loop) trace 
  is the trace of a balanced (resp.\ loop) derivation.
\end{definition}
We have several possibilities for characterizing a trace 
as balanced. An empty trace, denoted by $\epsilon$, 
is clearly balanced. 
For non-empty traces, we have several possibilities. 
In particular, a balanced trace cannot begin with the 
application of rules \textsc{Seq2}, \textsc{IfTrue2}, 
\textsc{IfFalse2}, or \textsc{LoopBaseS}, as these would 
produce an intermediate stack that does not extend 
the initial one. In summary, all balanced and loop
traces can be derived from the following grammar:
\[
\begin{array}{llll}
  \mathit{bal} & ::= & \epsilon \mid \textsc{AssVarS} 
  \mid \textsc{AssArrS}
  \mid \textsc{CallS}~\mathit{bal}
  \mid \textsc{UnCallS}~\mathit{bal}\\
  && \mid \textsc{Seq1}~\mathit{bal}~\textsc{Seq2}~\mathit{bal} \\
  &&\mid \textsc{IfTrue1}~\mathit{bal}~\textsc{IfTrue2}
  \mid \textsc{IfFalse1}~\mathit{bal}~\textsc{IfFalse2}\\
  &&\mid \textsc{LoopMainS}~\mathit{bal}~\mathit{loop}
  \\[1ex]
  \mathit{loop} & ::= & \textsc{LoopBaseS}
  \mid \textsc{Loop1}~\mathit{bal}~\textsc{Loop2}~\mathit{bal}~\mathit{loop}
\end{array}
\]
The following lemma states the soundness of the small-step
semantics.

\begin{lemma} \label{lemma:soundness}
  Consider a derivation of the form
  $\config{\sigma,s,\s_1} \to^\ast \config{\sigma',\ms{skip},\s_2}$.
  \begin{itemize}
  \item If the derivation is balanced, then $s$ is a proper 
  expression, $\s_1=\s_2$ and $\tuple{\sigma}{s} \Downarrow \sigma'$.
  \item If the derivation is a loop, then 
  $s = \ms{skip}$, $\s_1 = \ms{loop1}(e_1,s_1,e_2,s_2)\cons\s$,
  $\s_2 = \s$, and 
  $\tuple{\sigma}{(e_1,s_1,e_2,s_2)}\Downarrow \sigma'$.
  \end{itemize}
\end{lemma}

\begin{proof}
  We prove the claim by induction on length of the derivation,
  following the structure of balanced and loop
  traces, according to the grammar above:
  \begin{description}
  \item[\rm ($\epsilon$)] Then, $s$ must be $\ms{skip}$
  and the proof follows by applying rule \textsc{Skip} with
  $\sigma = \sigma'$.
  
  \item[\rm (\textsc{AssVarS})] Then, we have 
  $\config{\sigma,x~\oplus=e,\s} \rightarrow \config{\sigma[x \mapsto \llbracket \oplus \rrbracket(\sigma(x),v)],\ms{skip},\s}$, 
  with $\tuple{\sigma}{e} \Downarrow v$. Hence, the proof
  $\tuple{\sigma}{x~\oplus=e} \Downarrow 
  \sigma[x \mapsto \llbracket \oplus \rrbracket(\sigma(x),v)]$
  follows by applying rule \textsc{AssVar}.
  
  \item[\rm (\textsc{AssArrS})] This case is perfectly analogous
  to the previous one by applying rule \textsc{AssArr} instead.
  
  \item[\rm (\textsc{CallS}~\textit{bal})] Then, 
  $\config{\sigma,\ms{call}~id,\s}
  \rightarrow \config{\sigma,\Gamma(id),\s}
  \to^\ast \config{\sigma',\ms{skip},\s}$. 
  Now, since $\config{\sigma,\Gamma(id),\s} \to^\ast
  \config{\sigma',\ms{skip},\s}$ is balanced, by the ind.\
  hypothesis, we have
  $\tuple{\sigma}{\Gamma(id)} \Downarrow \sigma$.
  The claim thus follows by
  applying rule \textsc{Call} on 
  $\tuple{\sigma}{\ms{call}~\mathit{id}}$.
  
  \item[\rm (\textsc{UnCallS})] This case is perfectly analogous
  to the previous one by applying rule \textsc{UnCall} instead.

  \item[\rm (\textsc{Seq1}~\textit{bal}~\textsc{Seq2}~\textit{bal})]
  Then, we have a derivation of the following form:
  \[
  \begin{array}{lll}
  \config{\sigma,s_1~s_2,\s} &  \to & 
  \overbrace{\config{\sigma,s_1,\ms{seq}(s_2)\cons\s} \to^\ast \config{\sigma',\ms{skip},\ms{seq}(s_2)\cons\s}}^{(1)}\\
  & \to & \underbrace{\config{\sigma',s_2,\s}
  \to^\ast \config{\sigma'',\ms{skip},\s}}_{(2)}\\
  \end{array}
  \]
  Since both (1) and (2) above are balanced, by the inductive
  hypothesis, we have $\tuple{\sigma}{s_1}\Downarrow \sigma'$
  and $\tuple{\sigma'}{s_2}\Downarrow \sigma''$.
  Therefore, the proof follows by applying rule \textsc{Seq}
  on $\tuple{\sigma}{s_1~s_2}$.
  
  \item[\rm (\textsc{IfTrue1}~\textit{bal}~\textsc{IfTrue2})]
  Then, we have a derivation of the following form:
  \[
  \begin{array}{lll}
  \config{\sigma,\ms{if}~ e_1 ~\ms{then}~ s_1 ~\ms{else}~s_2~\ms{fi} ~e_2,\s} \\ 
  \to  \underbrace{\config{\sigma,s_1,\ms{if\_true}(e_2)\cons\s} 
   \to^\ast  \config{\sigma',\ms{skip},\ms{if\_true}(e_2)\cons\s}}_{(1)} 
  \to  \config{\sigma',\ms{skip},\s}
  \end{array}
  \]
  with $\tuple{\sigma}{e_1} \Downarrow	 v_1$, 
  $\ms{is\_true?}(v_1)$, 
  $\tuple{\sigma'}{e_2} \Downarrow v_2$, and 
  $\ms{is\_true?}(v_2)$. Since (1) above is balanced, by the
  inductive hypothesis, we have $\tuple{\sigma}{s_1} \Downarrow
  \sigma'$. Therefore, the proof follows by applying rule
  \textsc{IfTrue} on
  $\tuple{\sigma}{\ms{if}~ e_1 ~\ms{then}~ s_1 ~\ms{else}~s_2~\ms{fi} ~e_2}$.
  
  \item[\rm (\textsc{IfFalse1}~\textit{bal}~\textsc{IfFalse2})] 
  This case is perfectly analogous
  to the previous one by applying rule \textsc{IfFalse} instead.
  
  \item[\rm (\textsc{LoopMainS}~\textit{bal}~\textit{loop})]
  Then, we have a balanced derivation as follows:
  \[
  \begin{array}{lll}
  \config{\sigma,\ms{from}~ e_1 ~\ms{do}~ s_1 ~\ms{loop} ~s_2~\ms{until} ~e_2,\s} \\
  \to \underbrace{\config{\sigma,s_1,\ms{loop1}(e_1,s_1,e_2,s_2)\cons\s}
  \to^\ast \config{\sigma',\ms{skip},\ms{loop1}(e_1,s_1,e_2,s_2)\cons\s}}_{(1)} \\
  \to^\ast \config{\sigma'',\ms{skip},\s}
  \end{array}
  \]
  with $\tuple{\sigma}{e_1} \Downarrow v_1$ and  
  $\ms{is\_true?}(v_1)$.
  Since (1) above is balanced, by the ind.\ hypothesis, 
  we have $\tuple{\sigma}{s_1}\Downarrow\sigma'$. Moreover,
  since 
  $\config{\sigma',\ms{skip},\ms{loop1}(e_1,s_1,e_2,s_2)\cons\s}
  \to^\ast \config{\sigma'',\ms{skip},\s}$ is a \emph{loop} derivation,
  by the inductive hypothesis, we have
  $\tuple{\sigma'}{(e_1,s_1,e_2,s_2)}\Downarrow\sigma''$.
  Finally, the proof follows by applying rule \textsc{LoopMain}
  on $\tuple{\sigma}{\ms{from}~ e_1 ~\ms{do}~ s_1 ~\ms{loop} ~s_2~\ms{until} ~e_2}$.

  \item[\rm (\textsc{LoopBaseS})] Then, we have a one-step
  loop derivation of the following form:
  \[
  \config{\sigma,\ms{skip},\ms{loop1}(e_1,s_1,e_2,s_2)\cons\s}
  \to \config{\sigma,\ms{skip},\s}
  \]
  with 
  $\tuple{\sigma}{e_2} \Downarrow v_2$ and $\ms{is\_true?}(v_2)$.
  Trivially, we have $\tuple{\sigma}{(e_1,s_1,e_2,s_2)}
  \Downarrow \sigma$ by rule \textsc{LoopBase}.
  
  \item[\rm (\textsc{Loop1}~\textit{bal}~\textsc{Loop2}~\textit{bal}~\textit{loop})]
  Then, we have a loop derivation as follows:
  \[
  \begin{array}{lll}
  \config{\sigma,\ms{skip},\ms{loop1}(e_1,s_1,e_2,s_2)\cons\s} \\
  \to 
  \underbrace{\config{\sigma,s_2,\ms{loop2}(e_1,s_1,e_2,s_2)\cons\s} 
   \to^\ast \config{\sigma',\ms{skip},\ms{loop2}(e_1,s_1,e_2,s_2)\cons\s}}_{(1)} \\
  \to \underbrace{\config{\sigma',s_1,\ms{loop1}(e_1,s_1,e_2,s_2)\cons\s} 
  \to^\ast  \config{\sigma'',\ms{skip},\ms{loop1}(e_1,s_1,e_2,s_2)\cons\s}}_{(2)} \\
  \to^\ast \config{\sigma''',\ms{skip},\s}
  \end{array}
  \]
  with $\tuple{\sigma}{e_2} \Downarrow	 v_2$, 
  $\ms{is\_false?}(v_2)$, 
  $\tuple{\sigma'}{e_1} \Downarrow v_1$, and 
  $\ms{is\_false?}(v_1)$.
  Since both (1) and (2) above are balanced, by the
  inductive hypothesis, we have
  $\tuple{\sigma}{s_2}\Downarrow \sigma'$
  and $\tuple{\sigma'}{s_1}\Downarrow \sigma''$. 
  Moreover, since   
  $\config{\sigma'',\ms{skip},\ms{loop1}(e_1,s_1,e_2,s_2)\cons\s}
  \to^\ast \config{\sigma''',\ms{skip},\s}$ is a loop derivation,
  we have $\tuple{\sigma''}{(e_1,s_1,e_2,s_2)}\Downarrow \sigma'''$.
  Finally, the proof follows by applying rule \textsc{LoopRec} 
  on $\tuple{\sigma}{(e_1,s_1,e_2,s_2)}$. \qed
  \end{description}
\end{proof}
Now, we can proceed with the main result of correctness for
the small-step semantics:

\begin{theorem}\label{thm:equiv}
$
\tuple{\id}{s} \Downarrow  \sigma \text{ iff } 
\config{\id,s,\nul} \to^\ast \config{\sigma,\ms{skip},\nul}
$
\end{theorem}

\begin{proof}
  The ``only if'' part follows from Lemma~\ref{lemma:completeness}.
  The ``if'' part follows from Lemma~\ref{lemma:soundness} and
  the fact that any derivation of the form
  $\config{\id,s,\nul} \to^\ast \config{\sigma,\ms{skip},\nul}$
  is balanced. \qed
\end{proof}

\end{journal}


\section{A Reversible Semantics for Janus} \label{sec:reversible}

In this section, we propose a new reversible (small-step)
semantics based on the use of a \emph{program counter}.
Here, the configurations do not include a piece of code 
to be executed, as in the small-step semantics of the 
previous section, but a label that points to the next 
``block'' of the program to be executed.
This greatly simplifies the definition of a 
reversible semantics for the language.

Typically, semantics based on a program counter have been restricted
to simple, assembly-like languages, where execution is modeled by
simply incrementing a counter or updating it to perform jumps.  In
contrast, high-level language definitions rarely use an explicit
program counter, favoring instead the use of evaluation
contexts to manage control flow implicitly.  We refer to
related work (Section~\ref{sec:related}) for a more detailed
discussion on this topic.

\begin{figure}[t]
\begin{center}
$
\begin{array}{rcl@{~~~~~~}rcl@{ }}
\mi{p} & ::= & s ~(\ms{procedure} ~id~s)^{+} & \\
\mi{s} & ::= & \red{[}x \oplus = e\red{]^\ell} \, \mid \, 
\red{[}x[e]\oplus=e\red{]^\ell} \,
\mid 
\ms{if}~\red{[}e\red{]^\ell} ~\ms{then} ~s~\ms{else}~s~\ms{fi}~
\red{[}e\red{]^\ell} \\
& \mid & 
\ms{from}~\red{[}e\red{]^\ell}~\ms{do}~s~\ms{loop}~s~\ms{until}~\red{[}e\red{]^\ell} \,
\mid 
\red{[} \ms{call}~id\red{]^\ell} \, \mid \, 
\red{[} \ms{uncall}~id\red{]^\ell} \mid \, 
\red{[} \ms{skip}\red{]^\ell} \, \mid \, s~s\\
& \mid & \red{[\ms{start}]^\ell} \, \mid \, \red{[\ms{stop}]^\ell} 
\end{array}
$
\end{center}
\caption{Extended language syntax ($\mi{e}, \mi{c}, \oplus$ and $\odot$ are as in Fig.~\ref{syntax})} \label{fig:labeled-syntax}
\end{figure}

In the following, we introduce a control-flow graph (CFG)
to manage the evolution of the program counter. To this end, 
Fig.~\ref{fig:labeled-syntax} introduces an extended 
syntax where \emph{elementary blocks} (conditions, assignments, 
call, uncall, and $\ms{skip}$) are now labeled. 
We consider a domain of labels $\cL$, where 
$\ell,\ell_1,\ldots$ range over the set $\cL$.
All labels are unique in a program.
We let $\block(\ell) = o$
if $[o]^\ell$ belongs to the considered program,
where $o$ is either a statement or a condition.
\begin{conference}
\begin{figure}[t]
    \centering
    \fbox{
    \begin{minipage}{0.95\textwidth}
        \ttfamily 
        \begin{minipage}[t]{0.28\textwidth}
            [\kw{start}]$^1$ \\ \protect
            [n += 2]$^2$ \\
            {[\kw{call} sumMul2]}$^3$ \\
            {[\kw{stop}]}$^4$
        \end{minipage}
        \hfill 
        \begin{minipage}[t]{0.68\textwidth}
            \kw{procedure} sumMul2 \\
            \hspace*{1em} [\kw{start}]$^5$ \\
            \hspace*{1em} [i += 1]$^6$ \\
            \hspace*{1em} \kw{from} [i == 1]$^7$ \kw{do} \\
            \hspace*{2em} \kw{if} [(i \% 2) == 0]$^8$ \kw{then} [total += i]$^9$
            \hspace*{11.5em} \kw{else} [skip]$^{10}$ \\
            \hspace*{2em} \kw{fi} [(i \% 2) == 0]$^{11}$ \\
            \hspace*{1em} \kw{loop} \\
            \hspace*{2em} [i += 1]$^{12}$ \\
            \hspace*{1em} \kw{until} [i >= n]$^{13}$ \\
            \hspace*{1em} [n += total]$^{14}$\\
            \hspace*{1em} [\kw{stop}]$^{15}$\\
        \end{minipage}
    \end{minipage}
    }
    \caption{Sum2: a program computing the sum of the multiples of 2 up to $\tt n$}
    \label{fig:program}
\end{figure}
\end{conference}
\begin{journal}
\begin{figure}[t]
    \centering
    \fbox{
    \begin{minipage}{0.95\textwidth}
        \ttfamily 
        \begin{minipage}[t]{0.28\textwidth}
            [\kw{start}]$^1$ \\ \protect
            [n += 3]$^2$ \\
            {[\kw{call} sumMul3]}$^3$ \\
            {[\kw{stop}]}$^4$
        \end{minipage}
        \hfill 
        \begin{minipage}[t]{0.68\textwidth}
            \kw{procedure} sumMul3 \\
            \hspace*{1em} [\kw{start}]$^5$ \\
            \hspace*{1em} [i += 1]$^6$ \\
            \hspace*{1em} \kw{from} [i == 1]$^7$ \kw{do} \\
            \hspace*{2em} \kw{if} [(i \% 3) == 0]$^8$ \kw{then} [total += i]$^9$
            \hspace*{11.5em} \kw{else} [skip]$^{10}$ \\
            \hspace*{2em} \kw{fi} [(i \% 3) == 0]$^{11}$ \\
            \hspace*{1em} \kw{loop} \\
            \hspace*{2em} [i += 1]$^{12}$ \\
            \hspace*{1em} \kw{until} [i >= n]$^{13}$ \\
            \hspace*{1em} [n += total]$^{14}$\\
            \hspace*{1em} [\kw{stop}]$^{15}$\\
        \end{minipage}
    \end{minipage}
    }
    \caption{Sum3: a program computing the sum of the multiples of 3 up to $\tt n$}
    \label{fig:program}
\end{figure}
\end{journal}
Furthermore, we assume all programs (and procedures) have an
initial statement $\ms{start}$ and a 
final statement $\ms{stop}$ that simply mark the beginning 
and end (and is semantically equivalent to $\ms{skip}$).
\begin{conference}
A labeled Janus program, Sum2, is shown in Figure~\ref{fig:program}.
\end{conference}
\begin{journal}
A labeled Janus program, Sum3, is shown in Figure~\ref{fig:program}.
\end{journal}

\begin{figure}[t]
$
\begin{array}{rcl@{~~~~~~}rcl}
\entry(\ms{if}~[e_1]^{\ell_1} ~\ms{then} ~s_1~\ms{else}~s_2~\ms{fi}~
[e_2]^{\ell_2}) & = & \ell_1 & \entry(s_1~s_2)  & = &  \entry(s_1)\\
\entry(\ms{from}~[e_1]^{\ell_1}~\ms{do}~s_1~\ms{loop}~s_2~\ms{until}~[e_2]^{\ell_2}) & = & \ell_1 &
\entry([s]^\ell) & = & \ell ~\textrm{ otherwise}\\[1ex]
\exit(\ms{if}~[e_1]^{\ell_1} ~\ms{then} ~s_1~
\ms{else}~s_2~\ms{fi}~[e_2]^{\ell_2}) & = & \ell_2 &
\exit(s_1~s_2) & = & \exit(s_2) \\
\exit(\ms{from}~[e_1]^{\ell_1}~\ms{do}~s_1~
\ms{loop}~s_2~\ms{until}~[e_2]^{\ell_2}) & = & \ell_2 &
\exit([s]^\ell) & = & \ell ~\textrm{ otherwise}
\end{array}
$
\caption{Entry and exit points}
\label{fig:entryexit}
\end{figure}

Now, Figure~\ref{fig:entryexit} 
introduces the auxiliary function $\entry$
(resp.\ $\exit$), which returns the entry (resp.\ exit)
label of a program statement. Note that, in contrast
to most (high level) programming languages, $\exit$ always returns
a single label because of Janus conditionals with a single
exit point. The definitions are self-explanatory.
We also consider the following 
(partially defined) function,
$\context$, that identifies the statement where 
the label of a given expression (a test or assertion) occurs.
E.g., if the sentence 
``$\ms{if}~[e_1]^{\ell_1} ~\ms{then} ~s_1~\ms{else}~s_2~\ms{fi}~
[e_2]^{\ell_2}$'' occurs in a program, then
$\context(\ell_1) = \context(\ell_2) =
(\ms{if}~[e_1]^{\ell_1} ~\ms{then} ~s_1~\ms{else}~s_2~\ms{fi}~
[e_2]^{\ell_2})$.
Analogously, if the sentence 
``$\ms{from}~[e_1]^{\ell_1}~\ms{do}~s_1~\ms{loop}~s_2~\ms{until}~[e_2]^{\ell_2}$''
occurs in the program, then
$\context(\ell_1) = \context(\ell_2) =
(\ms{from}~[e_1]^{\ell_1}~\ms{do}~s_1~\ms{loop}~s_2~\ms{until}~[e_2]^{\ell_2})$.

\begin{figure}[t]
$
\begin{array}{rcl}
\flow(\ms{if}~[e_1]^{\ell_1} ~\ms{then} ~s_1~\ms{else}~s_2~\ms{fi}~
[e_2]^{\ell_2}) & = & \{(\ell_1,\entry(s_1)),(\exit(s_1),\ell_2)\}\\
& & \cup~ \{(\ell_1,\entry(s_2),(\exit(s_2),\ell_2)\}\\
& & \cup~ \flow(s_1)\cup\flow(s_2)\\
\flow(\ms{from}~[e_1]^{\ell_1}~\ms{do}~s_1~\ms{loop}~s_2~\ms{until}~[e_2]^{\ell_2}) & = & 
\{(\ell_1,\entry(s_1)),(\exit(s_1),\ell_2)\}\\
&  & \cup~ \{(\ell_2,\entry(s_2)),(\exit(s_2),\ell_1)\}\\
&  & 
\cup~\flow(s_1)\cup\flow(s_2)\\
\flow(s_1~s_2) & = & \{(\exit(s_1),\entry(s_2)\}  \cup  \flow(s_1) \cup\flow(s_2)\\
\flow(s) & = & \emptyset \quad \textrm{otherwise}
\end{array}
$
\caption{Control-flow graph of a Janus program} \label{fig:cfg}
\end{figure}

Then, function $\flow$ (Fig.~\ref{fig:cfg}) computes
the edges of the CFG of a Janus program.
It is worth noting that we do not have edges for
call/return (resp.\ uncall/return) in a CFG. 
The control flow for
these cases will be dealt with dynamically in
the corresponding transition rules.
We represent an edge $(\ell_1,\ell_2)$ graphically as 
$\ell_1 \edge \ell_2$.
Now, we proceed as follows:
\begin{itemize}
\item given a program, 
for each procedure ``$\ms{procedure}~\mathit{id}~s$'',
we introduce an inverse procedure of the form
``$\ms{procedure}~\mathit{id}^{-1}~\mathcal{I}\llbracket 
s \rrbracket$'', where the definition of function 
$\mathcal{I}\llbracket ~ \rrbracket$ from Fig.~\ref{fig:inverter} is extended so that
$\mathcal{I}\llbracket \ms{start} \rrbracket  ::=  \ms{stop}$
and $\mathcal{I}\llbracket \ms{stop} \rrbracket  ::=  \ms{start}$;
\item then, we add labels to each elementary block;
\item finally, we compute the CFGs of the main program
and of each procedure as well as its inverse version.
\end{itemize}
For example,
\begin{conference}
for program Sum2 (Fig.~\ref{fig:program}),
\end{conference}
\begin{journal}
for program Sum3 (Fig.~\ref{fig:program}),
\end{journal}
function $\mathit{flow}$ computes the following two 
CFGs (the graphical representation is shown in 
Fig.~\ref{fig:cfg-ex}):%
\begin{conference}
\footnote{For simplicity,
we do not show the CFG of $\mathit{sumMul2}^{-1}$ since 
there is no $\ms{uncall}$ in the program and, thus, 
it will not be used.}
\end{conference}
\begin{journal}
\footnote{For simplicity,
we do not show the CFG of $\mathit{sumMul3}^{-1}$ since 
there is no $\ms{uncall}$ in the program and, thus, 
it will not be used.}
\end{journal}
\[
\begin{array}{rl}
CFG_{\mathit{main}}  = \{ & (1,2), (2,3),
(3, 4) ~\}\\[1ex]
CFG_{\mathit{sumMul3}} = \{ & 
(5,6), 
(6,7), 
(7,8), 
(8,9), (8,10), 
(9,11), (10,11),\\
& (11, 13), 
(13, 12), (12,7), 
(13,14), (14,15)~\}
\end{array}
\]

\begin{figure}[t]
\centering
\begin{tikzpicture}[
    node distance=0.5cm and 0.5cm,
    state/.style={
        circle, 
        draw, 
        minimum size=0.5cm,
        inner sep=0pt,
        font=\footnotesize
    },
    true/.style={
        draw=black, 
        thin, 
        solid, 
        -Latex
    },
    label_text/.style={
        font=\itshape\footnotesize, 
        anchor=east
    }
]

    \node[state] (1) at (0,0) {1};
    \node[state] (2) [right=of 1] {2};
    \node[state] (3) [right=of 2] {3};
    \node[state] (4) [right=of 3] {4};
    \node[label_text] at ([xshift=-0.2cm]1.west) {Main program};
    \draw[true] (1) -- (2);
    \draw[true] (2) -- (3);
    \draw[true] (3) -- (4);

    \node[state] (5) [below=0.8cm of 1] {5}; 
    \node[state] (6) [right=of 5] {6};
    \node[state] (7) [right=of 6] {7}; 
    \node[state] (8) [right=of 7] {8}; 
    \node[label_text] at ([xshift=-0.2cm]5.west) {Procedure $\mathit{sumMul3}$};
    \node[state] (9) [above right=0.1cm and 0.5cm of 8] {9};
    \node[state] (10) [below right=0.1cm and 0.5cm of 8] {10};
    \node[state] (11) [right=1.8cm of 8] {11};
    \node[state] (13) [right=of 11] {13};
    \node[state] (12) [below=0.4cm of 13] {12};
    \node[state] (14) [right=of 13] {14};
    \node[state] (15) [right=of 14] {15};
    \draw[true] (5) -- (6);
    \draw[true] (6) -- (7);
    \draw[true] (7) -- (8);
    \draw[true] (8) -- (9);
    \draw[true] (8) -- (10);
    \draw[true] (9) -- (11);
    \draw[true] (10) -- (11);
    \draw[true] (11) -- (13);
    \draw[true] (13) -- (14);
    \draw[true] (14) -- (15);
    \draw[true] (13) -- (12);
    \draw[true] (12) -| (7);
\end{tikzpicture}
\caption{CFGs for program Sum3} \label{fig:cfg-ex}
\end{figure}
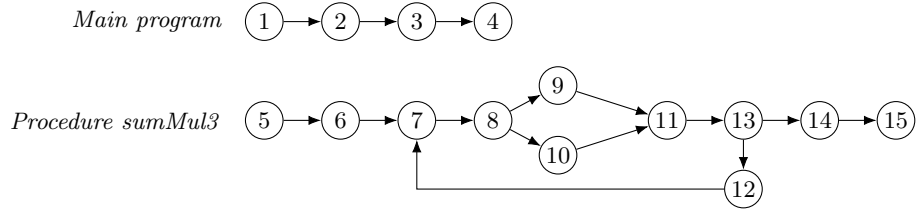

\subsection{The Forward Rules}

\begin{figure}[tp]
\begin{mathpar}
\inferrule*[leftstyle=\rm,left=$\overrightarrow{\mathit{AssVar}}$] 
{\block(\ell') = (x~\oplus=e) \\  
\tuple{\sigma}{e} \Downarrow v \\
\ell' \edge \ell'' } 
{ \config{\sigma,\ell,\ell',\s} \rh \config{\sigma[x \mapsto \llbracket \oplus \rrbracket(\sigma(x),v)],\ell',\ell'',\s} }     
\and
\inferrule*[leftstyle=\rm,left=$\overrightarrow{\mathit{AssArr}}$] 
{ 
\block(\ell') = (x[e_l]~\oplus=e) \\
\tuple{\sigma}{ e_l} \Downarrow	 v_l \\
\tuple{\sigma}{ e} \Downarrow	 v\\
\ell' \edge \ell'' 
} 
{
\config{\sigma,\ell,\ell',\s} \rh 
\config{\sigma[x[v_l] \mapsto \llbracket \oplus \rrbracket(\sigma(x[v_l]),v)],\ell',\ell'',\s} 
}
\and
\inferrule*[leftstyle=\rm,left=$\overrightarrow{\mathit{Call}}$]
{\block(\ell') = (\ms{call}~id)\\
 \entry(\Gamma(id)) \edge \ell''           
}           
{ \config{\sigma,\ell,\ell',\s} \rh  \config{\sigma,\entry(\Gamma(id)),\ell'',\ms{call}(\ell')\cons\s}  }
\and
\inferrule*[leftstyle=\rm,left=$\overrightarrow{\mathit{Return1}}$] 
{\block(\ell') = \ms{stop}\\
\ell'' \edge \ell'''
}
{ \config{\sigma,\ell,\ell',\ms{call}(\ell'')\cons\s} 
\rh  \config{\sigma,\ell'',\ell''',\s}  }
\and
\inferrule*[leftstyle=\rm,left=$\overrightarrow{\mathit{UnCall}}$]
{\block(\ell') = (\ms{uncall}~id)\\
 \entry(\Gamma(id^{-1})) \edge \ell''           
}           
{ \config{\sigma,\ell,\ell',\s} \rh  \config{\sigma,\entry(\Gamma(id^{-1})),\ell'',\ms{uncall}(\ell')\cons\s}  }
\and
\inferrule*[leftstyle=\rm,left=$\overrightarrow{\mathit{Return2}}$] 
{\block(\ell') = \ms{stop}\\
\ell'' \edge \ell'''
}
{ \config{\sigma,\ell,\ell',\ms{uncall}(\ell'')\cons\s} 
\rh  \config{\sigma,\ell'',\ell''',\s}  }
\and
\inferrule*[leftstyle=\rm,left=$\overrightarrow{\mathit{Skip}}$] 
{
\block(\ell') = \ms{skip}\\
\ell' \edge \ell''
}
{ 
\config{\sigma,\ell,\ell',\s} \rh  \config{\sigma,\ell',\ell'',\s}  
}
\end{mathpar}
\caption{Reversible forward semantics: rules for basic constructs} \label{fig:reversible-smallstep}
\end{figure}

Now, we introduce the transition relation, $\rh$,
defining the 
forward direction of the reversible semantics.
Here, configurations have the form 
$\config{\sigma,\ell,\ell',\s}$, where the control of
the configuration is determined by a pair of labels 
representing the \emph{last executed statement} and the 
\emph{next statement to be executed}, respectively. 
Actually, the transition rules of the forward semantics 
only use the latter. 
We include the former 
in the configurations since it will be used by the transition
rules of the backward semantics, and in this way we have 
a single definition of configuration that works for both 
directions.
The rules defining the basic constructs can be found in
Fig.~\ref{fig:reversible-smallstep}. These rules are
quite similar to their counterpart in the small-step
semantics of Fig.~\ref{fig:smallstep}. The main differences
are the following:
\begin{itemize}
\item First, configurations no longer contain the next 
statement to be executed, but rather a label that 
identifies it (acting as a \emph{program counter}). 
Thus, we use the auxiliary function $\block$ to identify 
the block pointed to by that label. 
Furthermore, in many cases, the rules increment the 
program counter according to the edges of the CFG.

\item Rule $\overrightarrow{\mathit{Call}}$ updates the 
control of the configuration with the labels from the 
first two statements of the procedure. Specifically, 
the first label will point to the statement $\ms{start}$
(which is not executed), so the second label will 
point to the actual first statement of the procedure.
Rule $\overrightarrow{Return1}$, on the other hand, 
is triggered when the next statement to be executed 
is $\ms{stop}$. In this case, we update the control of
the configuration so that the last executed statement 
is $\ms{call}$, and the next statement to be executed
is the one following it (according to the CFG).

\item Rules $\overrightarrow{\mathit{UnCall}}$ and 
$\overrightarrow{\mathit{Return2}}$ proceed analogously, 
with the only difference being that, in this case, 
we consider the inverse procedure.

\item Finally, we introduce a new rule, 
$\overrightarrow{\mathit{Skip}}$, to deal with
$\ms{skip}$ which just increments the program counter 
(according to the CFG).
\end{itemize}
\begin{figure}[tp]
\begin{mathpar}
\inferrule*[leftstyle=\rm,left=$\overrightarrow{\mathit{IfTrue1}}$] 
{ 
\block(\ell_1) = e_1\\ 
\context(\ell_1) = (\ms{if}~ [e_1]^{\ell_1} ~\ms{then}~ s_1 ~\ms{else}~s_2~\ms{fi} ~[e_2]^{\ell_2})
\\
\tuple{\sigma}{e_1} \Downarrow	 v_1 \\ 
\ms{is\_true?}(v_1) 
} 
{ \config{\sigma,\ell,\ell_1,\s} \rh  \config{\sigma,\ell_1,\entry(s_1),
    \ms{if\_true}(\ell_1,\ell_2)\cons\s}  
}
\and 
\inferrule*[leftstyle=\rm,left=$\overrightarrow{\mathit{IfTrue2}}$] 
{ 
\block(\ell_2) = e_2\\
\tuple{\sigma}{e_2} \Downarrow	 v_2 \\ 
\ms{is\_true?}(v_2) \\
\ell_2 \edge \ell'_2
}
{ \config{\sigma,\ell,\ell_2,\ms{if\_true}(\ell_1,\ell_2)\cons\s}
\rh 
\config{\sigma,\ell_2,\ell'_2,\s}
}
\and 
\inferrule*[leftstyle=\rm,left=$\overrightarrow{\mathit{IfFalse1}}$] 
{ 
\block(\ell_1) = e_1\\ 
\context(\ell_1) = (\ms{if}~ [e_1]^{\ell_1} ~\ms{then}~ s_1 ~\ms{else}~s_2~\ms{fi} ~[e_2]^{\ell_2})
\\
\tuple{\sigma}{e_1} \Downarrow	 v_1 \\ 
\ms{is\_false?}(v_1) 
} 
{ \config{\sigma,\ell,\ell_1,\s} \rh  \config{\sigma,\ell_1,\entry(s_2),\ms{if\_false}(\ell_1,\ell_2)\cons\s}  
}
\and 
\inferrule*[leftstyle=\rm,left=$\overrightarrow{\mathit{IfFalse2}}$] 
{ 
\block(\ell_2) = e_2\\
\tuple{\sigma}{e_2} \Downarrow	 v_2 \\ 
\ms{is\_false?}(v_2) \\
\ell_2 \edge \ell'_2
}
{ \config{\sigma,\ell,\ell_2,\ms{if\_false}(\ell_1,\ell_2)\cons\s}  
\rh 
\config{\sigma,\ell_2,\ell'_2,\s}
}
\and
\inferrule*[leftstyle=\rm,left=$\overrightarrow{\mathit{LoopMain}}$] 
{ 
\block(\ell_1) = e_1\\
\context(\ell_1) =
(\ms{from}~ [e_1]^{\ell_1} ~\ms{do}~ s_1 ~\ms{loop}~s_2~\ms{until} ~[e_2]^{\ell_2})\\
\tuple{\sigma}{e_1} \Downarrow	 v_1\\ 
\ms{is\_true?}(v_1) 
} 
{
\config{\sigma,\ell,\ell_1,\s}
\rh  
\config{\sigma,\ell_1,\entry(s_1),\ms{loop1}(\ell_1,s_1,\ell_2,s_2)\cons\s}  
}
\and
\inferrule*[leftstyle=\rm,left=$\overrightarrow{\mathit{LoopBase}}$] 
{ 
\block(\ell_2) = e_2 \\
\tuple{\sigma}{e_2} \Downarrow	 v_2\\ 
\ms{is\_true?}(v_2) \\
\ell_2 \edge \ell'_2 \\
\ell'_2 \neq \entry(s_1)
} 
{
\config{\sigma,\ell,\ell_2,\ms{loop1}(\ell_1,s_1,\ell_2,s_2)\cons\s}
\rh
\config{\sigma,\ell_2,\ell'_2,\s}  
}
\and
\inferrule*[leftstyle=\rm,left=$\overrightarrow{\mathit{Loop1}}$] 
{ 
\block(\ell_2) = e_2 \\
\tuple{\sigma}{e_2} \Downarrow	 v_2\\ 
\ms{is\_false?}(v_2) 
} 
{
\config{\sigma,\ell,\ell_2,\ms{loop1}(\ell_1,s_1,\ell_2,s_2)\cons\s}  
\rh
\config{\sigma,\ell_2,\entry(s_2),\ms{loop2}(\ell_1,s_1,\ell_2,s_2)\cons\s}
}
\and
\inferrule*[leftstyle=\rm,left=$\overrightarrow{\mathit{Loop2}}$] 
{ 
\block(\ell_1) = e_1\\
\tuple{\sigma}{e_1} \Downarrow	 v_1\\ 
\ms{is\_false?}(v_1) 
} 
{
\config{\sigma,\ell,\ell_1,\ms{loop2}(\ell_1,s_1,\ell_2,s_2)\cons\s}  
\rh
\config{\sigma,\ell_1,\entry(s_1),\ms{loop1}(\ell_1,s_1,\ell_2,s_2)\cons\s} 
}
\end{mathpar}
\caption{Reversible forward semantics: rules for if and for loop} \label{fig:reversible-ifloop}
\end{figure}
The second group of rules, the rules for conditional and loop
statements, can be seen in Figure~\ref{fig:reversible-ifloop}.
These rules are perfectly analogous to their 
counterpart in the small-step semantics of Fig.~\ref{fig:ifloop}.
The only relevant difference lies in the use of a 
program counter to identify the next block to be executed. 
Note also that the use of the auxiliary function 
$\context$ is necessary to determine whether the Boolean 
expression is part of a conditional statement or of a loop.

On the other hand, observe that the stack elements $\ms{if\_true}$ and 
$\ms{if\_false}$ contain labels for both the test and 
the assertion of the conditional. In practice, only the 
second label is required (by either rule $\overrightarrow{\mathit{IfTrue2}}$ and $\overrightarrow{\mathit{IfFalse2}}$). 
However, the rules of the backward semantics will require 
the first label (that of the test). Hence, 
we include both labels in the elements of the stack 
so that they can be used by both the forward and 
the backward semantics.

Also, when the program counter points to the test of the loop ($\ell_2$), 
we can either terminate the loop and proceed to the next 
statement or execute $s_2$ and return to the 
initial assertion ($\ell_1$). Consequently, $\ell_2$ has 
two outgoing edges in the CFG. Therefore, we add 
the condition $\ell'_2 \neq \entry(s_1)$ in rule 
$\overrightarrow{\mathit{LoopBase}}$ in order to 
ensure that we are considering the right edge.

Finally, the elements of the stack, $\ms{loop1}$ and 
$\ms{loop2}$, could in principle contain only four labels, 
$(\ell_1,\entry(s_1),\ell_2,\entry(s_2))$. 
Indeed, this would be enough for the forward semantics. 
However, in the backward semantics, we will need to access 
$\exit(s_1)$ and $\exit(s_2)$. Therefore, we
use $(\ell_1,s_1,\ell_2,s_2)$ as arguments 
of $\ms{loop1}$ and $\ms{loop2}$, so that they contain 
enough information for both semantics.

Given a program $p = (s~p_1~\ldots~p_n)$, an \emph{initial
configuration} has now the form 
$\config{\emptystate,\entry(s),\ell,\nul}$, where
$\entry(s) \edge \ell$ belongs to the CFG of the program.

Now, we will state the equivalence between 
the small-step semantics of the previous section
and the reversible (forward) semantics 
just introduced.
\begin{journal}
In what follows, we assume that program blocks are also 
labeled when applying the rules of small-step semantics
of the previous section, which simply ignore the annotations.
Furthermore, for simplicity, although we consider a single
program in our correctness result, we assume that
the additional statements $\ms{start}$ and $\ms{stop}$
only occur when considering the reversible semantics.
\end{journal}

\newcommand{\eqconf}{\sim_{\mathit{conf}}}
\newcommand{\eqstack}{\sim_{\mathit{stack}}}

\begin{lemma} \label{lemma:forward-rev}
  Let $s$ be a statement where $\entry(s) \edge \ell$ 
  belongs to the CFG.
  Then, $\config{\emptystate,s,\nul} \to^\ast 
  \config{\sigma,\ms{skip},\nul}$
  iff  $\config{\emptystate,\entry(s),\ell,\nul} 
  \rh^\ast 
  \config{\sigma,\ell',\ell'',\nul}$ with 
  $\block(\ell'') = \ms{stop}$.
\end{lemma}
\begin{journal}
\begin{figure}[t]
	\begin{mathpar}
		\inferrule* 
		{
		} 
		{ 
			\nul \eqstack \nul
		}
		\and
		\inferrule* 
		{
			\s \eqstack \s'
		} 
		{ 
			\ms{seq}(s)\cons\s \eqstack \s'
		}
		\and
		\inferrule* 
		{
			\s \eqstack \s'
		} 
		{ 
			\s \eqstack \ms{call}(\ell)\cons\s'
		}
		\and
		\inferrule* 
		{
			\s \eqstack \s'
		} 
		{ 
			\s \eqstack \ms{uncall}(\ell)\cons\s'
		}
		\and
		\inferrule* 
		{
			\context(\ell_2) = (\ms{if}~ [e_1]^{\ell_1} ~\ms{then}~ s_1 ~\ms{else}~s_2~\ms{fi} ~[e_2]^{\ell_2})
			\\
			\s \eqstack \s'
		} 
		{ 
			\ms{if\_true}([e_2]^{\ell_2})\cons\s \eqstack \ms{if\_true}(\ell_1,\ell_2)\cons\s'
		}
		\and
		\inferrule* 
		{
			\context(\ell_2) = (\ms{if}~ [e_1]^{\ell_1} ~\ms{then}~ s_1 ~\ms{else}~s_2~\ms{fi} ~[e_2]^{\ell_2})
             \\
			\s \eqstack \s'
		} 
		{ 
			\ms{if\_false}([e_2]^{\ell_2})\cons\s \eqstack \ms{if\_false}(\ell_1,\ell_2)\cons\s'
		}
		\and
		\inferrule* 
		{
			\s \eqstack \s'
		}
		{ 
			\ms{loop1}([e_1]^{\ell_1},s_1,[e_2]^{\ell_2},s_2)\cons\s \eqstack \ms{loop1}(\ell_1,s_1,\ell_2,s_2)\cons\s'
		}
		\and
		\inferrule* 
		{
			\s \eqstack \s'
		}
		{ 
			\ms{loop2}([e_1]^{\ell_1},s_1,[e_2]^{\ell_2},s_2)\cons\s \eqstack \ms{loop2}(\ell_1,s_1,\ell_2,s_2)\cons\s'
		}
	\end{mathpar}
	\caption{Equivalence relation on stacks ``$\eqstack$''}
	\label{fig:eqstack}
\end{figure}
In order to prove this result, we begin by introducing an 
equivalence relation, $\eqstack$, between the stacks of the small-step 
semantics configurations and those of the reversible semantics
(shown in Fig.~\ref{fig:eqstack}).

Moreover, we also introduce the following auxiliary function,
$\cont$, that returns the label of the \emph{first} 
block to be executed 
once the current statement in the control of the 
configuration (of the small-step semantics)
is completed:
\[
\begin{array}{rcl}
\cont(\ms{seq}(s)\cons\s) & = & \entry(s) \\
\cont(\ms{if\_true}([e]^\ell)\cons\s) & = & \ell \\
\cont(\ms{if\_false}([e]^\ell)\cons\s) & = & \ell \\
\cont(\ms{loop1}([e_1]^{\ell_1},s_1,[e_2]^{\ell_2},s_2)\cons\s) & = & 
\ell_2 \\
\cont(\ms{loop2}([e_1]^{\ell_1},s_1,[e_2]^{\ell_2},s_2)\cons\s) & = & 
\ell_1 \\
\end{array}
\]
In the following, to distinguish them, we denote the configurations
of the small-step semantics of the previous section with 
$c, c', c_1, \ldots$ and the configurations of the reversible 
semantics with $k,k',k_1,\ldots$

We say that two configurations, 
$c = \config{\sigma,s,\s}$ and $k=\config{\sigma',\ell_1,\ell_2,\s'}$, are \emph{equivalent}, denoted $c\eqconf k$, if and only if
\begin{enumerate}
 \item $\sigma = \sigma'$,
 \item $\s \eqstack \s'$, and
 \item either $\ell_2 = \entry(s)$, or $s=\ms{skip}$ and $\ell_2 = \cont(\s)$.
\end{enumerate} 
Note that $\ell_1$ is not relevant when considering only the
forward rules.

Let us now proceed with the proof of Lemma~\ref{lemma:forward-rev}.

\begin{proof}
  $(\Rightarrow)$ We consider two configurations, 
  $c_1$ and $k_1$, with $c_1 \eqconf k_1$. 
  We will show that if $c_1 \to c_2$, 
  then $k_1 \rh^\ast k_2$ and $c_2 \eqconf k_2$, since
  the claim follows trivially from this result.
  We distinguish the following cases depending on the rule applied:
  \begin{description}
    \item[\rm (\textsc{AssVarS})] 
    Then $c_1 = \config{\sigma, [x~\oplus=e]^{\ell'},\s} \to
    \config{\sigma[x \mapsto \llbracket \oplus \rrbracket(\sigma(x),v),\ms{skip},\s} = c_2$, 
    with $\tuple{\sigma}{e} \Downarrow v$. 
    Since $c_1\eqconf k_1$, we have 
    $k_1 = \config{\sigma,\ell,\ell',\s'}$ 
    with $\s \eqstack \s'$. By rule $\overrightarrow{\mathit{AssVar}}$,
    we have $k_2 = \config{\sigma[x \mapsto \llbracket \oplus \rrbracket(\sigma(x),v),\ell',\ell'',\s'}$, with $\ell'\edge\ell''$.
    Since the top element of the stack $\s$ contains the next 
    statement to
    be executed, and assuming that the CFG correctly represents 
    the control flow of the program, we have $\ell'' = \cont(\s)$ 
    and, thus, $c_2 \eqconf k_2$.
    
    \item[\rm (\textsc{AssArrS})] This case is perfectly analogous
    to the previous one.
    
    \item[\rm (\textsc{CallS})] Here, we assume that
    $\Gamma(\mathit{id})$ returns the body of procedure $\mathit{id}$,
    as usual, and $\Gamma'(\mathit{id})$ returns the same statement
    but excluding the additional statements $\ms{start}$ and
    $\ms{stop}$ that are not considered by the small-step semantics.
    Then $c_1 = \config{\sigma, [\ms{call}~\mathit{id}]^{\ell'},\s}
    \to \config{\sigma,\Gamma'(\mathit{id}),\s} = c_2$. 
    Since $c_1\eqconf k_1$, we have 
    $k_1 = \config{\sigma,\ell,\ell',\s'}$ 
    with $\s \eqstack \s'$. By rule $\overrightarrow{\mathit{Call}}$,
    we have $k_2 = \config{\sigma,\entry(\Gamma(\mathit{id})),\ell'',\ms{call}(\ell')\cons\s'}$, with $\entry(\Gamma(\mathit{id}))\edge\ell''$.
    Trivially, $\ell''$ points to the first statement in $\Gamma(\mathit{id})$ which is different from $\ms{start}$, i.e.,
    $\ell'' = \entry(\Gamma'(\mathit{id}))$.
    Moreover, since $\ms{call}$ elements do not affect to the
    stack equivalence, we have $\s \eqstack \ms{call}(\ell')\cons\s'$
    and, thus, $c_2 \eqconf k_2$.
    
    \item[\rm (\textsc{UnCallS})] This case is perfectly analogous
    to the previous one.
    
    \item[\rm (\textsc{Seq1})] Then $c_1 = \config{\sigma, s_1~s_2,\s}
    \to \config{\sigma,s_1,\ms{seq}(s_2)\cons\s} = c_2$. 
    Since $c_1\eqconf k_1$, we have 
    $k_1 = \config{\sigma,\ell,\ell',\s'}$ 
    with $\ell' = \entry(s_1~s_2) = \entry(s_1)$ and $\s \eqstack \s'$. 
    Hence, the claim follows by applying zero steps in the
    reversible semantics since 
    $\ms{seq}(s_2)\cons\s \eqstack \s'$ and, thus,
    $c_2 \eqconf k_1 = k_2$.
    
    \item[\rm (\textsc{Seq2})] Then 
    $c_1 = \config{\sigma, \ms{skip},\ms{seq}(s_2)\cons\s}\to
    \config{\sigma,s_2,\s} = c_2$. 
    Since $c_1\eqconf k_1$, we have 
    $k_1 = \config{\sigma,\ell,\ell',\s'}$ 
    with $\ell' = \entry(s_2)$ and $\ms{seq}(s_2)\cons\s 
    \eqstack \s'$ (thus, $\s\eqconf \s'$). 
    Hence, the claim follows by applying zero steps in the
    reversible semantics since $c_2 \eqconf k_1 = k_2$.
    
    \item[\rm (\textsc{IfTrue1})] Then  
    $c_1 = \config{\sigma, \ms{if}~ [e_1]^{\ell_1} ~\ms{then}~ s_1 ~\ms{else}~s_2~\ms{fi} ~[e_2]^{\ell_2},\s}
    \to \config{\sigma,s_1,\ms{if\_true}([e_2]^{\ell_2})\cons\s} = c_2$
    with $\tuple{\sigma}{e_1} \Downarrow v_1$ and  
    $\ms{is\_true?}(v_1)$.
    Since $c_1\eqconf k_1$, we have 
    $k_1 = \config{\sigma,\ell,\ell_1,\s'}$ 
    with $\s \eqstack \s'$. 
    By applying rule $\overrightarrow{\mathit{IfTrue1}}$,
    we have $k_2 = \config{\sigma,\ell_2,\entry(s_1),
    \ms{if\_true}(\ell_1,\ell_2)\cons\s'}$. Hence,
    $c_2\eqconf k_2$ follows trivially since
    $\ms{if\_true}([e_2]^{\ell_2})\cons\s \eqstack \ms{if\_true}(\ell_1,\ell_2)\cons\s'$.
    
    \item[\rm (\textsc{IfTrue2})] Then $c_1 = \config{\sigma, \ms{skip},\ms{if\_true}([e_2]^{\ell_2})\cons\s}
    \to \config{\sigma,\ms{skip},\s} = c_2$. 
    Since $c_1\eqconf k_1$, we have 
    $k_1 = \config{\sigma,\ell,\ell_2,\ms{if\_true}(\ell_1,\ell_2)\cons\s'}$ 
    with $\s \eqstack \s'$. 
    By rule $\overrightarrow{\mathit{IfTrue2}}$, we have
    $k_2 = \config{\sigma,\ell_2,\ell_2',\s'}$ with 
    $\ell_2 \edge \ell'_2$. Since the top element of the 
    stack $\s$ contains the next statement to
    be executed, assuming that the CFG correctly represents 
    the control flow of the program, we have $\ell'_2 = \cont(\s)$ 
    and, thus, $c_2 \eqconf k_2$.
    
    \item[\rm (\textsc{IfFalse1}) and (\textsc{IfFalse2})] These
    cases are perfectly analogous to the previous two cases.
    
    \item[\rm (\textsc{LoopMainS})] Then, we have 
    $c_1 = \config{\sigma, \ms{from}~ [e_1]^{\ell_1} ~\ms{do}~ s_1 ~\ms{loop}~s_2~\ms{until} ~[e_2]^{\ell_2}),\s}
    \to \config{\sigma,s_1,\ms{loop1}([e_1]^{\ell_1},s_2,[e_2]^{\ell_2},s_2)\cons\s} = c_2$
    with $\tuple{\sigma}{e_1} \Downarrow v_1$ and  
    $\ms{is\_true?}(v_1)$. 
    Since $c_1\eqconf k_1$, we have 
    $k_1 = \config{\sigma,\ell,\ell_1,\s'}$ 
    with $\s \eqstack \s'$. 
    By rule $\overrightarrow{\mathit{LoopMain}}$,
    we have $k_2 = \config{\sigma,\ell_1,\entry(s_1),
    \ms{loop1}(\ell_1,s_1,\ell_2,s_2)\cons\s'}$. Hence,
    $\ms{loop1}([e_1]^{\ell_1},s_1,[e_2]^{\ell_2},s_2)\cons\s \eqstack \ms{loop1}(\ell_1,s_1,\ell_2,s_2)\cons\s'$ and, thus,
    $c_2\eqconf k_2$.
    
    \item[\rm (\textsc{LoopBaseS})] Then $c_1 = \config{\sigma, \ms{skip},\ms{loop1}([e_1]^{\ell_1},s_1,[e_2]^{\ell_2},s_2)\cons\s}
    \to \config{\sigma,\ms{skip},\s} = c_2$,
    with $\tuple{\sigma}{e_2} \Downarrow v_2$ and  
    $\ms{is\_true?}(v_2)$.
    Since $c_1\eqconf k_1$, we have 
    $k_1 = \config{\sigma,\ell,\ell_2,\ms{loop1}(\ell_1,s_1,\ell_2,s_2)\cons\s'}$ 
    and, moreover,  $\ms{loop1}([e_1]^{\ell_1},s_1,[e_2]^{\ell_2},s_2)\cons\s \eqstack \ms{loop1}(\ell_1,s_1,\ell_2,s_2)\cons\s'$ (and
    thus $\s \eqstack \s'$). 
    By rule $\overrightarrow{\mathit{LoopBase}}$, we have
    $k_2 = \config{\sigma,\ell_2,\ell_2',\s'}$ with 
    $\ell_2 \edge \ell'_2$ and $\ell'_2\neq\entry(s_1)$. 
    Since the top element of the 
    stack $\s$ contains the next statement to
    be executed once the loop is completed, assuming that 
    the CFG correctly represents 
    the control flow of the program, we have $\ell'_2 = \cont(\s)$ 
    and, thus, $c_2 \eqconf k_2$.
    
    \item[\rm (\textsc{Loop1})]  Then, we have the step  $c_1 = \config{\sigma, \ms{skip},\ms{loop1}([e_1]^{\ell_1},s_1,[e_2]^{\ell_2},s_2)\cons\s}
    \to \config{\sigma,s_2,\ms{loop2}([e_1]^{\ell_1},s_1,[e_2]^{\ell_2},s_2)\cons\s} = c_2$, 
    with $\tuple{\sigma}{e_2} \Downarrow v_2$ and  
    $\ms{is\_false?}(v_2)$.
    Since $c_1\eqconf k_1$, we have 
    $k_1 = \config{\sigma,\ell,\ell_2,\ms{loop1}(\ell_1,s_1,\ell_2,s_2)\cons\s'}$, 
    where $\ms{loop1}([e_1]^{\ell_1},s_1,[e_2]^{\ell_2},s_2)\cons\s \eqstack \ms{loop1}(\ell_1,s_1,\ell_2,s_2)\cons\s'$.
    By rule $\overrightarrow{\mathit{Loop1}}$, we have
    $k_2 = \config{\sigma,\ell_2,\entry(s_2),\ms{loop2}(\ell_1,s_1,\ell_2,s_2)\cons\s'}$ and $c_2 \eqconf k_2$ holds.
    
    \item[\rm (\textsc{Loop2})]  Then, we have the step $c_1 = \config{\sigma, \ms{skip},\ms{loop2}([e_1]^{\ell_1},s_1,[e_2]^{\ell_2},s_2)\cons\s}
    \to \config{\sigma,s_1,\ms{loop1}([e_1]^{\ell_1},s_1,[e_2]^{\ell_2},s_2)\cons\s} = c_2$,
    with $\tuple{\sigma}{e_1} \Downarrow v_1$ and  
    $\ms{is\_false?}(v_1)$. 
    Since $c_1\eqconf k_1$, we have 
    $k_1 = \config{\sigma,\ell,\ell_1,\ms{loop2}(\ell_1,s_1,\ell_2,s_2)\cons\s'}$, where 
    $\ms{loop2}([e_1]^{\ell_1},s_1,[e_2]^{\ell_2},s_2)\cons\s \eqstack \ms{loop2}(\ell_1,s_1,\ell_2,s_2)\cons\s'$.
    By rule $\overrightarrow{\mathit{Loop2}}$, we have
    $k_2 = \config{\sigma,\ell_1,\entry(s_1),\ms{loop1}(\ell_1,s_1,\ell_2,s_2)\cons\s'}$ and $c_2 \eqconf k_2$ holds.
  \end{description} 
  ($\Leftarrow$) Again, we consider two configurations, 
  $c_1$ and $k_1$, with $c_1 \eqconf k_1$. 
  Now, we will show that if $k_1 \rh k_2$, 
  then $c_1 \to^\ast c_2$ and $c_2 \eqconf k_2$, since
  the claim follows trivially from this result.
  We distinguish the following cases depending on the rule applied:
  \begin{description}
  \item[\rm ($\overrightarrow{\mathit{AssVar}}$)] Then, 
  $k_1 = \config{\sigma,\ell,\ell',\s'}\rh \config{\sigma[x \mapsto \llbracket \oplus \rrbracket(\sigma(x),v)],\ell',\ell'',\s'} = k_2$, with
  $\block(\ell') = (x~\oplus=e)$, 
  $\tuple{\sigma}{e} \Downarrow v$, and $\ell' \edge \ell''$.
  Since $c_1 \eqconf k_1$, we have that 
  $c_1 = \config{\sigma,s,\s}$ where $\s \eqstack \s'$ and, thus,
  $s$ can be either $[x~\oplus=e]^{\ell'}$,
  a sequence whose first statement is $[x~\oplus=e]^{\ell'}$,
  or $\ms{skip}$ with $\ell' = \cont(\s)$. 
  W.l.o.g, let us
  consider the second case (the other two cases would be similar).
  Then, we have a derivation of the form
  $c_1 \to^\ast \config{\sigma,x~\oplus=e,\s_2} = c'_1$ such
  that $\s_2 \eqstack \s'$ since only $\ms{seq}$ elements
  have been added to the stack. Now, by rule 
  \textsc{AssVarS}, we have
  $c'_1 \to \config{\sigma[x \mapsto \llbracket \oplus \rrbracket(\sigma(x),v)],\ms{skip},\s_2}$ and the claim follows
  since the top element of the stack $\s_2$ contains the next 
  statement to be executed and, assuming that the CFG 
  correctly represents the control flow of the program, 
  we have $\ell'' = \cont(\s_2)$.
  
  \item[\rm ($\overrightarrow{\mathit{AssArrS}}$)] This case is
  perfectly analogous to the previous one.
  
  \item[\rm ($\overrightarrow{\mathit{Call}}$)] As in the proof
  above, we assume that $\Gamma'(\mathit{id})$ returns the 
  body of procedure $\mathit{id}$ excluding the 
  statements $\ms{start}$ and $\ms{stop}$.
  Then, we have 
  $k_1 = \config{\sigma,\ell,\ell',\s'} \rh  \config{\sigma,\entry(\Gamma(id)),\ell'',\ms{call}(\ell')\cons\s'} = k_2$,
  with $\block(\ell') = (\ms{call}~id)$ and 
  $\entry(\Gamma(id)) \edge \ell''$.
  Since $c_1 \eqconf k_1$, we have that 
  $c_1 = \config{\sigma,\ms{call}~\mathit{id},\s}$ with 
  $\s \eqstack \s'$.\footnote{For simplicity, in this case 
  and subsequent ones, we will ignore the possibility of the 
  configuration having a sequence in the control. 
  Otherwise, we would have to proceed as in the first case
  and perform a number of steps with rules \textsc{Seq1} and/or 
  \textsc{Seq2}.}
  By rule \textsc{CallS}, we have 
  $c_1 \to \config{\sigma,\Gamma'(id),\s} = c_2$.
  Trivially, $\ell''$ points to the first statement in 
  $\Gamma(\mathit{id})$ which is different from 
  $\ms{start}$, i.e., $\ell''$ points to 
  $\entry(\Gamma'(\mathit{id}))$.
  Moreover, since $\ms{call}$ elements do not affect to the
  stack equivalence, we have $c_2 \eqconf k_2$.
  
  \item[\rm ($\overrightarrow{\mathit{Return1}}$)] Then, we
  have $k_1 = \config{\sigma,\ell,\ell',\ms{call}(\ell'')\cons\s'} 
  \rh  \config{\sigma,\ell'',\ell''',\s'} = k_2$,
  with $\block(\ell') = \ms{stop}$ and $\ell'' \edge \ell'''$
  in the CFG.
  Then, since $c_1 \eqconf k_1$, we have 
  $c_1 = \config{\sigma,\ms{skip},\s}$ with $\s \eqstack \s'$
  and $\cont(\s)$ pointing to the next statement after the call.
  In the small-step semantics, there is no 
  rule to model the return of a procedure since, once the 
  execution of the procedure body is completed, 
  execution simply continues with the next statement after 
  the corresponding $\ms{call}$. Assuming that $\ms{call}$ was
  part of a sequence (otherwise, execution is done), 
  $\s$ will have the form $\ms{seq}(s)\cons\s_2$, where $s$ 
  is the next statement to be executed after the $\ms{call}$. 
  Therefore, applying rule \textsc{Seq2}, we have
  $c_1 \to \config{\sigma,s,\s_2}$ and the claim follows 
  since $\ell''' =\entry(s)$, and $\s_2 \eqstack \s'$ also
  holds since $\ms{seq}$ elements 
  are not considered by the relation ``$\eqstack$''.
  
  \item[\rm ($\overrightarrow{\mathit{UnCall}}$ and $\overrightarrow{\mathit{Return2}}$)] These cases are perfectly
  analogous to the previous two cases.

  \item[\rm ($\overrightarrow{\mathit{Skip}}$)] Then, we have
  $k_1 = \config{\sigma,\ell,\ell',\s'} \rh  \config{\sigma,\ell',\ell'',\s'} = k_2$, 
  with $\block(\ell') = \ms{skip}$ and $\ell' \edge \ell''$ in
  the CFG.
  Since $c_1 = k_1$, we have $c_1 = \config{\sigma,\ms{skip},\s}$
  with $\s \eqstack \s'$.  
  We note that rule $\overrightarrow{\mathit{Skip}}$ is needed in 
  the reversible semantics only in case there is a statement 
  $\ms{skip}$ in the program. 
  In contrast, within the small-step semantics, 
  a similar rule is not necessary since it will simply be 
  dealt with as any other ``finished'' statement 
  (the execution of all statements ends with $\ms{skip}$). 
  Therefore, we consider a zero-step derivation with 
  $c_1 = c_2$, and $c_2 \eqconf k_2$ holds since $c_1 \eqconf k_1$
  and, thus, $\ell' = \ell'' = \cont(\s)$.
  
  \item[\rm ($\overrightarrow{\mathit{IfTrue1}}$] Then, 
  $k_1 = \config{\sigma,\ell,\ell_1,\s'} \rh  
  \config{\sigma,\ell_1,\entry(s_1),\ms{if\_true}(\ell_1,\ell_2)\cons\s'} = k_2$, with $\block(\ell_1) = e_1$,  
  $\context(\ell_1) = (\ms{if}~ [e_1]^{\ell_1} ~\ms{then}~ s_1 ~\ms{else}~s_2~\ms{fi} ~[e_2]^{\ell_2})$,
  $\tuple{\sigma}{e_1} \Downarrow v_1$, and 
  $\ms{is\_true?}(v_1)$.
  Since $c_1 \eqconf k_1$ holds, we have that
  $c_1 = \config{\sigma,\ms{if}~ [e_1]^{\ell_1} ~\ms{then}~ s_1 ~\ms{else}~s_2~\ms{fi} ~[e_2]^{\ell_2},\s}$ with $\s \eqstack \s'$.
  By rule \textsc{IfTrue1}, we derive
  $c_2 = \config{\sigma,s_1,\ms{if\_true}(e_2)\cons\s}$ and
  the claim follows.
  
  \item[\rm ($\overrightarrow{\mathit{IfTrue2}}$)] Then, we have
  $k_1 = \config{\sigma,\ell,\ell_2,\ms{if\_true}(\ell_1,\ell_2)\cons\s'}
  \rh \config{\sigma,\ell_2,\ell'_2,\s'} = k_2$, with 
  $\block(\ell_2) = e_2$, $\tuple{\sigma}{e_2} \Downarrow v_2$,
  $\ms{is\_true?}(v_2)$, and $\ell_2 \edge \ell'_2$.
  Since $c_1 \eqconf k_1$, we have that 
  $c_1 = \config{\sigma,\ms{skip},\ms{if\_true}([e_2]^{\ell_2})\cons\s}$ 
  with $\s \eqstack \s'$. By rule \textsc{IfTrue2}, we have
  $c_1 \to \config{\sigma,\ms{skip},\s} = c_2$. Hence, 
  assuming that the CFG correctly represents the control 
  flow of the program, we have that $\ell'_2 = \cont(\s)$
  points to the next statement in the stack to be 
  executed after the conditional and, thus,
  $c_2 \eqconf k_2$.
  
  \item[\rm ($\overrightarrow{\mathit{IfFalse1}}$ and $\overrightarrow{\mathit{IfFalse2}}$)] These cases are perfectly
  analogous to the previous two cases.
  
  \item[\rm ($\overrightarrow{\mathit{LoopMain}}$)] Then, $k_1 = \config{\sigma,\ell,\ell_1,\s'}
  \rh  
  \config{\sigma,\ell_1,\entry(s_1),\ms{loop1}(\ell_1,s_1,\ell_2,s_2)\cons\s'} = k_2$, with  
  $\block(\ell_1) = e_1$, $\context(\ell_1) =
  (\ms{from}~ [e_1]^{\ell_1} ~\ms{do}~ s_1 ~\ms{loop}~s_2~
  \ms{until} ~[e_2]^{\ell_2})$, 
  $\tuple{\sigma}{e_1} \Downarrow v_1$, and
  $\ms{is\_true?}(v_1)$. 
  Since $c_1 \eqconf k_1$ holds, we have that
  $c_1 = \config{\sigma,\ms{from}~ [e_1]^{\ell_1} ~\ms{do}~ s_1 ~\ms{loop}~s_2~
  \ms{until} ~[e_2]^{\ell_2},\s}$ with $\s \eqstack \s'$.
  By rule \textsc{LoopMainS}, we derive
  $c_2 = \config{\sigma,s_1,\ms{loop1}([e_1]^{\ell_1},s_1,[e_2]^{\ell_2},s_2)\cons\s}$ and
  the claim follows.
  
  \item[\rm ($\overrightarrow{\mathit{LoopBase}}$)] Then, 
  $k_1 = \config{\sigma,\ell,\ell_2,\ms{loop1}(\ell_1,s_1,\ell_2,s_2)\cons\s'}
  \rh \config{\sigma,\ell_2,\ell'_2,\s'} = k_2$, with 
  $\block(\ell_2) = e_2$, $\tuple{\sigma}{e_2} \Downarrow v_2$,
  $\ms{is\_true?}(v_2)$, $\ell_2 \edge \ell'_2$ in the CFG, and 
  $\ell'_2 \neq \entry(s_1)$.
  Since $c_1 \eqconf k_1$, we have  
  $c_1 = \config{\sigma,\ms{skip},\ms{loop1}([e_1]^{\ell_1},s_1,[e_2]^{\ell_2},s_2)\cons\s}$ 
  with $\s \eqstack \s'$. By rule \textsc{LoopBaseS}, we have
  $c_1 \to \config{\sigma,\ms{skip},\s} = c_2$. Hence, 
  assuming that the CFG correctly represents the control 
  flow of the program, we have that $\ell'_2 = \cont(\s)$
  points to the next statement in the stack to be 
  executed after the loop and, thus, $c_2 \eqconf k_2$.
  
  \item[\rm ($\overrightarrow{\mathit{Loop1}}$)] Then, we have the following step
  $k_1 = \config{\sigma,\ell,\ell_2,\ms{loop1}(\ell_1,s_1,\ell_2,s_2)\cons\s'}
  \rh 
  \config{\sigma,\ell_2,\entry(s_2),\ms{loop2}(\ell_1,s_1,\ell_2,s_2)\cons\s'} = k_2$, 
  with 
  $\block(\ell_2) = e_2$, $\tuple{\sigma}{e_2} \Downarrow v_2$,
  and $\ms{is\_false?}(v_2)$.
  Since $c_1 \eqconf k_1$ holds, we have that
  $c_1 = \config{\sigma,\ms{skip},\ms{loop1}([e_1]^{\ell_1},s_1,[e_2]^{\ell_2},s_2)\cons\s}$ with $\s \eqstack \s'$.
  By rule \textsc{Loop1}, we derive
  $c_2 = \config{\sigma,s_2,\ms{loop2}([e_1]^{\ell_1},s_1,[e_2]^{\ell_2},s_2)\cons\s}$ and
  the claim follows.
  
  \item[\rm ($\overrightarrow{\mathit{Loop2}}$)] Then, we have the following step
  $k_1 = \config{\sigma,\ell,\ell_1,\ms{loop1}(\ell_1,s_1,\ell_2,s_2)\cons\s'}
  \rh 
  \config{\sigma,\ell_1,\entry(s_1),\ms{loop1}(\ell_1,s_1,\ell_2,s_2)\cons\s'} = k_2$, 
  with 
  $\block(\ell_1) = e_1$, $\tuple{\sigma}{e_1} \Downarrow v_1$,
  and $\ms{is\_false?}(v_1)$.
  Since $c_1 \eqconf k_1$ holds, we have that
  $c_1 = \config{\sigma,\ms{skip},\ms{loop2}([e_1]^{\ell_1},s_1,[e_2]^{\ell_2},s_2)\cons\s}$ with $\s \eqstack \s'$.
  By rule \textsc{Loop2}, we derive
  $c_2 = \config{\sigma,s_1,\ms{loop1}([e_1]^{\ell_1},s_1,[e_2]^{\ell_2},s_2)\cons\s}$ and
  the claim follows. \qed
  \end{description}
\end{proof}
\end{journal}

\begin{conference}
\begin{example} \label{ex:forward_trace}
  Given the program in Fig.~\ref{fig:program}
  (and the CFG in Fig.~\ref{fig:cfg-ex}), we have the
  forward derivation shown in 
  Figure~\ref{fig:forward_derivation}, where we highlight 
  in red the ``active'' label in each step (i.e., the label 
  that represents the next statement to be executed by the
  forward reversible semantics).
  
  \begin{figure}
  $
  \begin{array}{rl}
  & \config{\emptystate, 1,\red{2}, \nul}\\
  
  {}_{\overrightarrow{\mathit{AssVar}}} &  
  \config{\{n\mapsto 2\},2,\red{3},\nul}\\

  {}_{\overrightarrow{\mathit{Call}}} & 
  \config{\{n\mapsto 2\},5,\red{6},\ms{call}(3)\cons\nul}\\

  {}_{\overrightarrow{\mathit{AssVar}}} & 
  \config{\{n\mapsto 2,i\mapsto 1\},6,\red{7},\ms{call}(3)\cons\nul}\\

  {}_{\overrightarrow{\mathit{LoopMain}}} & 
  \config{\{n\mapsto 2, i\mapsto 1\},7,\red{8},\ms{loop1}(7,s_1,13,s_2)\cons\ms{call}(3)\cons\nul}\\

  {}_{\overrightarrow{\mathit{IfFalse1}}} & 
  \config{\{n\mapsto 2,i\mapsto 1\},8,\red{10},\ms{if\_false}(8,11)\cons\ms{loop1}(7,s_1,13,s_2)\cons\ms{call}(3)\cons\nul}\\

  {}_{\overrightarrow{\mathit{Skip}}} & 
  \config{\{n\mapsto 2,i\mapsto 1\},10,\red{11},\ms{if\_false}(8,11)\cons\ms{loop1}(7,s_1,13,s_2)\cons\ms{call}(3)\cons\nul}\\

  {}_{\overrightarrow{\mathit{IfFalse2}}} & 
  \config{\{n\mapsto 2,i\mapsto 1\},11,\red{13},\ms{loop1}(7,s_1,13,s_2)\cons\ms{call}(3)\cons\nul}\\
  
  {}_{\overrightarrow{\mathit{Loop1}}} & 
  \config{\{n\mapsto 2,i\mapsto 1\},13,\red{12},\ms{loop2}(7,s_1,13,s_2)\cons\ms{call}(3)\cons\nul}\\
  
  {}_{\overrightarrow{\mathit{AssVar}}} & 
  \config{\{n\mapsto 2,i\mapsto 2\},12,\red{7},\ms{loop2}(7,s_1,13,s_2)\cons\ms{call}(3)\cons\nul}\\

  {}_{\overrightarrow{\mathit{Loop2}}} & 
  \config{\{n\mapsto 2,i\mapsto 2\},7,\red{8},\ms{loop1}(7,s_1,13,s_2)\cons\ms{call}(3)\cons\nul}\\

%
%
%
%
%

  {}_{\overrightarrow{\mathit{IfTrue1}}} & 
  \config{\{n\mapsto 2,i\mapsto 2\},8,\red{9},\ms{if\_true}(8,11)\cons\ms{loop1}(7,s_1,13,s_2)\cons\ms{call}(3)\cons\nul}\\
  
  {}_{\overrightarrow{\mathit{AssVar}}} & 
  \config{\{n\mapsto 2,i\mapsto 2,total\mapsto 2\},9,\red{11},\ms{if\_true}(8,11)\cons\ms{loop1}(7,s_1,13,s_2)\cons\ms{call}(3)\cons\nul}\\

  {}_{\overrightarrow{\mathit{IfTrue2}}} & 
  \config{\{n\mapsto 2,i\mapsto 2,total\mapsto 2\},11,\red{13},\ms{loop1}(7,s_1,13,s_2)\cons\ms{call}(3)\cons\nul}\\

  {}_{\overrightarrow{\mathit{LoopBase}}} & 
  \config{\{n\mapsto 2,i\mapsto 2,total\mapsto 2\},13,\red{14},\ms{call}(3)\cons\nul}\\

  {}_{\overrightarrow{\mathit{AssVar}}} & 
  \config{\{n\mapsto 4,i\mapsto 2,total\mapsto 2\},14,\red{15},\ms{call}(3)\cons\nul}\\
  
  {}_{\overrightarrow{\mathit{Return1}}} & 
  \config{\{n\mapsto 4,i\mapsto 2,total\mapsto 2\},3,\red{4},\nul}\\
  \end{array}
  $
  \caption{Forward derivation with the reversible semantics}
  \label{fig:forward_derivation}
  \end{figure}
\end{example}
\end{conference}

\begin{journal}
\begin{example} \label{ex:forward_trace}
  Given the program in Fig.~\ref{fig:program}
  (and the CFG in Fig.~\ref{fig:cfg-ex}), we have the
  forward derivation shown in 
  Figure~\ref{fig:forward_derivation}, where we highlight 
  in red the ``active'' label in each step (i.e., the label 
  that represents the next statement to be executed by the
  forward reversible semantics).
  
  \begin{figure}
  $
  \begin{array}{rl}
  & \config{\emptystate, 1,\red{2}, \nul}\\
  
  {}_{\overrightarrow{\mathit{AssVar}}} &  
  \config{\{n\mapsto 3\},2,\red{3},\nul}\\

  {}_{\overrightarrow{\mathit{Call}}} & 
  \config{\{n\mapsto 3\},5,\red{6},\ms{call}(3)\cons\nul}\\

  {}_{\overrightarrow{\mathit{AssVar}}} & 
  \config{\{n\mapsto 3,i\mapsto 1\},6,\red{7},\ms{call}(3)\cons\nul}\\

  {}_{\overrightarrow{\mathit{LoopMain}}} & 
  \config{\{n\mapsto 3, i\mapsto 1\},7,\red{8},\ms{loop1}(7,s_1,13,s_2)\cons\ms{call}(3)\cons\nul}\\

  {}_{\overrightarrow{\mathit{IfFalse1}}} & 
  \config{\{n\mapsto 3,i\mapsto 1\},8,\red{10},\ms{if\_false}(8,11)\cons\ms{loop1}(7,s_1,13,s_2)\cons\ms{call}(3)\cons\nul}\\

  {}_{\overrightarrow{\mathit{Skip}}} & 
  \config{\{n\mapsto 3,i\mapsto 1\},10,\red{11},\ms{if\_false}(8,11)\cons\ms{loop1}(7,s_1,13,s_2)\cons\ms{call}(3)\cons\nul}\\

  {}_{\overrightarrow{\mathit{IfFalse2}}} & 
  \config{\{n\mapsto 3,i\mapsto 1\},11,\red{13},\ms{loop1}(7,s_1,13,s_2)\cons\ms{call}(3)\cons\nul}\\
  
  {}_{\overrightarrow{\mathit{Loop1}}} & 
  \config{\{n\mapsto 3,i\mapsto 1\},13,\red{12},\ms{loop2}(7,s_1,13,s_2)\cons\ms{call}(3)\cons\nul}\\
  
  {}_{\overrightarrow{\mathit{AssVar}}} & 
  \config{\{n\mapsto 3,i\mapsto 2\},12,\red{7},\ms{loop2}(7,s_1,13,s_2)\cons\ms{call}(3)\cons\nul}\\

  {}_{\overrightarrow{\mathit{Loop2}}} & 
  \config{\{n\mapsto 3,i\mapsto 2\},7,\red{8},\ms{loop1}(7,s_1,13,s_2)\cons\ms{call}(3)\cons\nul}\\

  {}_{\overrightarrow{\mathit{IfFalse1}}} & 
  \config{\{n\mapsto 3,i\mapsto 2\},8,\red{10},\ms{if\_false}(8,11)\cons\ms{loop1}(7,s_1,13,s_2)\cons\ms{call}(3)\cons\nul}\\

  {}_{\overrightarrow{\mathit{Skip}}} & 
  \config{\{n\mapsto 3,i\mapsto 2\},10,\red{11},\ms{if\_false}(8,11)\cons\ms{loop1}(7,s_1,13,s_2)\cons\ms{call}(3)\cons\nul}\\

  {}_{\overrightarrow{\mathit{IfFalse2}}} & 
  \config{\{n\mapsto 3,i\mapsto 2\},11,\red{13},\ms{loop1}(7,s_1,13,s_2)\cons\ms{call}(3)\cons\nul}\\

  {}_{\overrightarrow{\mathit{Loop1}}} & 
  \config{\{n\mapsto 3,i\mapsto 2\},13,\red{12},\ms{loop2}(7,s_1,13,s_2)\cons\ms{call}(3)\cons\nul}\\
  
  {}_{\overrightarrow{\mathit{AssVar}}} & 
  \config{\{n\mapsto 3,i\mapsto 3\},12,\red{7},\ms{loop2}(7,s_1,13,s_2)\cons\ms{call}(3)\cons\nul}\\
  
  {}_{\overrightarrow{\mathit{Loop2}}} & 
  \config{\{n\mapsto 3,i\mapsto 3\},7,\red{8},\ms{loop1}(7,s_1,13,s_2)\cons\ms{call}(3)\cons\nul}\\

  {}_{\overrightarrow{\mathit{IfTrue1}}} & 
  \config{\{n\mapsto 3,i\mapsto 3\},8,\red{9},\ms{if\_true}(8,11)\cons\ms{loop1}(7,s_1,13,s_2)\cons\ms{call}(3)\cons\nul}\\
  
  {}_{\overrightarrow{\mathit{AssVar}}} & 
  \config{\{n\mapsto 3,i\mapsto 3,total\mapsto 3\},9,\red{11},\ms{if\_true}(8,11)\cons\ms{loop1}(7,s_1,13,s_2)\cons\ms{call}(3)\cons\nul}\\

  {}_{\overrightarrow{\mathit{IfTrue2}}} & 
  \config{\{n\mapsto 3,i\mapsto 3,total\mapsto 3\},11,\red{13},\ms{loop1}(7,s_1,13,s_2)\cons\ms{call}(3)\cons\nul}\\

  {}_{\overrightarrow{\mathit{LoopBase}}} & 
  \config{\{n\mapsto 3,i\mapsto 3,total\mapsto 3\},13,\red{14},\ms{call}(3)\cons\nul}\\

  {}_{\overrightarrow{\mathit{AssVar}}} & 
  \config{\{n\mapsto 6,i\mapsto 3,total\mapsto 3\},14,\red{15},\ms{call}(3)\cons\nul}\\
  
  {}_{\overrightarrow{\mathit{Return1}}} & 
  \config{\{n\mapsto 6,i\mapsto 3,total\mapsto 3\},3,\red{4},\nul}\\
  \end{array}
  $
  \caption{Forward derivation with the reversible semantics}
  \label{fig:forward_derivation}
  \end{figure}
\end{example}
\end{journal}

\subsection{The Backward Rules}

This section introduces the rules of the \emph{backward} 
reversible semantics. Basically, a backward rule operates 
as follows: it first identifies the statement referenced 
by the \emph{first} label within the configuration control 
(i.e., the label corresponding to the last executed statement), 
then undoes its effects, and finally updates the 
labels accordingly (mainly traversing the original CFG backwards).

In the following,
$\flow^{-1}$ denotes the function that computes the \emph{inverse} CFG
of a program: 
$
\flow^{-1}(s) = \{ (\ell',\ell) \mid (\ell,\ell')\in\flow(s)\}
$.
For clarity, in the transition rules of the backward semantics, we 
denote the edges $(\ell,\ell')\in \flow^{-1}(s)$ of the inverse 
CFG as $\ell \iedge \ell'$; i.e., $\ell \iedge \ell'$ is 
equivalent to $\ell' \edge \ell$ in the original CFG.

The following result states the equivalence between computing
the inverse CFG of a statement and a standard CFG of the
corresponding inverse statement:

\begin{lemma} Given a statement $s$, we have
$\flow^{-1}(s) = \flow(\mathcal{I}\llbracket s \rrbracket)$.
\end{lemma}

\begin{journal}
\begin{proof}
  We prove the claim by structural induction on the considered
  statement.
  For elementary blocks, the proof follows trivially. Let us now
  consider the remaining cases:
  \begin{itemize}
  \item Let $s = (\ms{if}~ [e_1]^{\ell_1} ~\ms{then}~ s_1 
~\ms{else}~s_2~\ms{fi} ~[e_2]^{\ell_2})$. Then, we have
  \[
  \begin{array}{llll}
  \flow(\mathcal{I}\llbracket s \rrbracket) & = & 
  \flow(\ms{if}~[e_2]^{\ell_2} ~\ms{then} ~\mathcal{I}\llbracket s_1 \rrbracket~\ms{else}~\mathcal{I}\llbracket s_2 \rrbracket~\ms{fi}~[e_1]^{\ell_1}) \\
  
  & = & \{(\ell_2,\entry(\mathcal{I}\llbracket s_1 \rrbracket)),
  (\exit(\mathcal{I}\llbracket s_1 \rrbracket),\ell_1),\\
  && ~~(\ell_2,\entry(\mathcal{I}\llbracket s_2 \rrbracket)),
  (\exit(\mathcal{I}\llbracket s_2 \rrbracket),\ell_1)\}\\
  &&~~\cup \flow(\mathcal{I}\llbracket s_1 \rrbracket) 
  \cup \flow(\mathcal{I}\llbracket s_2 \rrbracket) \\
  
  & = & \{(\ell_2,\exit(s_1)),
  (\entry(s_1),\ell_1),(\ell_2,\exit(s_2)),
  (\entry(s_2),\ell_1)\}\\
  &&~~\cup \flow(\mathcal{I}\llbracket s_1 \rrbracket) 
  \cup \flow(\mathcal{I}\llbracket s_2 \rrbracket) \\

  & = & \{(\ell_2,\exit(s_1)),
  (\entry(s_1),\ell_1),(\ell_2,\exit(s_2)),
  (\entry(s_2),\ell_1)\}\\
  &&~~\cup \flow^{-1}(s_1) 
  \cup \flow^{-1}(s_2) ~~~\mbox{(by ind.\ hypothesis)}\\
  
  & = & \flow^{-1}(s)
  \end{array}
  \]
  assuming that $\entry(\mathcal{I}\llbracket s \rrbracket) = 
  \exit(s)$ and that $\exit(\mathcal{I}\llbracket s \rrbracket) =
  \entry(s)$, which is straightforward.
  
  \item Let $s = (\ms{from}~ [e_1]^{\ell_1} ~\ms{do}~ s_1 ~\ms{loop}~s_2~\ms{until} ~[e_2]^{\ell_2})$. Then, we have
  \[
  \begin{array}{llll}
  \flow(\mathcal{I}\llbracket s \rrbracket) & = & 
  \flow(\ms{from}~[e_2]^{\ell_2}~\ms{do}~\mathcal{I}\llbracket s_1 \rrbracket~\ms{loop}~\mathcal{I}\llbracket s_2 \rrbracket~\ms{until}~[e_1]^{\ell_1}) \\
  
  & = & \{(\ell_2,\entry(\mathcal{I}\llbracket s_1 \rrbracket)), 
  (\exit(\mathcal{I}\llbracket s_1 \rrbracket),\ell_1),\\
  &&~~(\ell_1,\entry(\mathcal{I}\llbracket s_2 \rrbracket)),
  (\exit(\mathcal{I}\llbracket s_2 \rrbracket),\ell_2)\}\\
  && ~~\cup\flow(\mathcal{I}\llbracket s_1 \rrbracket)
  \cup \flow(\mathcal{I}\llbracket s_2 \rrbracket) \\
  
  & = & \{(\ell_2,\exit(s_1)), 
  (\entry(s_1),\ell_1), (\ell_1,\exit(s_2)),
  (\entry(s_2),\ell_2)\}\\
  && ~~\cup\flow(\mathcal{I}\llbracket s_1 \rrbracket)
  \cup \flow(\mathcal{I}\llbracket s_2 \rrbracket) \\
  
  & = & \{(\ell_2,\exit(s_1)), 
  (\entry(s_1),\ell_1), (\ell_1,\exit(s_2)),
  (\entry(s_2),\ell_2)\}\\
  && ~~\cup\flow^{-1}(s_1)
  \cup \flow^{-1}(s_2) ~~~\mbox{(by ind.\ hypothesis)}\\

  & = & \flow^{-1}(s) 
  \end{array}
  \]
  
  \item Finally, let $s = (s_1~s_2)$. Then, we have
  \[
  \begin{array}{llll}  
  \flow(\mathcal{I}\llbracket s_1~s_2 \rrbracket)
  & = & \flow(\mathcal{I}\llbracket s_2 \rrbracket~
  \mathcal{I}\llbracket s_1 \rrbracket) \\

  & = & \{(\exit(\mathcal{I}\llbracket s_2 \rrbracket),\entry(\mathcal{I}\llbracket s_1 \rrbracket))\} \cup\flow(\mathcal{I}\llbracket s_2 \rrbracket)\cup\flow(\mathcal{I}\llbracket s_1 \rrbracket)\\
  
  & = & \{(\entry(s_2),\exit(s_1))\} \cup\flow(\mathcal{I}\llbracket s_2 \rrbracket)\cup\flow(\mathcal{I}\llbracket s_1 \rrbracket)\\
  
  & = & \{(\entry(s_2),\exit(s_1))\} \cup\flow^{-1}(s_2)\cup\flow^{-1}(s_1) ~~~\mbox{(by IH)}\\
  
  & = & \flow^{-1}(s_1~s_2) \\
  \end{array} 
  \]
  \qed
  \end{itemize}
\end{proof}
\end{journal}

\begin{figure}[tp]
  \begin{mathpar}
        \inferrule*[leftstyle=\rm,left=$\overleftarrow{\mathit{AssVar}}$] 
{
\block(\ell') = (x~\oplus=e) \\  
\tuple{\sigma}{e} \Downarrow v \\
\ell'' \iedge \ell' 
} 
{ 
\config{\sigma,\ell',\ell'',\s} \lh \config{\sigma[x \mapsto \mathcal{I}_{op}\llbracket \oplus \rrbracket(\sigma(x),v)],\ell,\ell',\s}
}     
\and
\inferrule*[leftstyle=\rm,left=$\overleftarrow{\mathit{AssArr}}$] 
{ 
\block(\ell') = (x[e_l]~\oplus=e) \\
\tuple{\sigma}{ e_l} \Downarrow	 v_l \\
\tuple{\sigma}{ e} \Downarrow	 v\\
\ell'' \iedge \ell'
} 
{
\config{\sigma,\ell',\ell'',\s} \lh 
\config{\sigma[x[v_l] \mapsto \mathcal{I}_{op}\llbracket \oplus \rrbracket(\sigma(x[v_l]),v)],\ell,\ell',\s}
}     
\and
\inferrule*[leftstyle=\rm,left=$\overleftarrow{\mathit{Call}}$] 
{\block(\ell_s) = \ms{start}\\
\ell' \iedge \ell             
} 
{
\config{\sigma,\ell_s,\ell'',\ms{call}(\ell')\cons\s}
\lh
\config{\sigma,\ell,\ell',\s}
}
\and
\inferrule*[leftstyle=\rm,left=$\overleftarrow{\mathit{Return1}}$] 
{
  \block(\ell'') = (\ms{call}~id)\\
  \exit(\Gamma(\mathit{id})) \iedge \ell
}
{ 
\config{\sigma,\ell'',\ell''',\s} \lh
\config{\sigma,\ell,\exit(\Gamma(\mathit{id})),\ms{call}(\ell'')\cons\s} 
}
\and
\inferrule*[leftstyle=\rm,left=$\overleftarrow{\mathit{UnCall}}$] 
{\block(\ell_s) = \ms{start}\\
\ell' \iedge \ell             
} 
{
\config{\sigma,\ell_s,\ell'',\ms{uncall}(\ell')\cons\s}
\lh
\config{\sigma,\ell,\ell',\s}
}
\and
\inferrule*[leftstyle=\rm,left=$\overleftarrow{\mathit{Return2}}$] 
{
  \block(\ell'') = (\ms{uncall}~id)\\
  \exit(\Gamma(\mathit{id}^{-1})) \iedge \ell
}
{ 
\config{\sigma,\ell'',\ell''',\s} \lh
\config{\sigma,\ell,\exit(\Gamma(\mathit{id}^{-1})),\ms{uncall}(\ell'')\cons\s} 
}
\and
\inferrule*[leftstyle=\rm,left=$\overleftarrow{\mathit{Skip}}$] 
{
\block(\ell') = \ms{skip}\\
\ell'' \iedge \ell' \\
}
{ 
\config{\sigma,\ell',\ell'',\s}  \lh \config{\sigma,\ell,\ell',\s}  
}
\end{mathpar}
\caption{Reversible backward semantics: rules for basic constructs} \label{fig:reversible-smallstep-backward}
\end{figure}

Let us now present the transition relation, $\lh$,
defining the 
backward direction of the reversible semantics.
The backward rules for basic constructs are shown
in Figure~\ref{fig:reversible-smallstep-backward}.
In principle, each rule $\overleftarrow{\mathit{Rule}}$ 
represents the inverse version of the corresponding 
forward rule $\overrightarrow{{Rule}}$ from 
Fig.~\ref{fig:reversible-smallstep}, 
executing the inverse of the last executed statement 
(pointed to by the first label in the control of the configuration) 
and advancing in the inverse CFG (equivalently, moving backwards 
in the standard CFG). 
In particular, note that the rules 
$\overleftarrow{\mathit{Call}}$ and 
$\overleftarrow{\mathit{Uncall}}$ are triggered upon 
reaching the statement $\ms{start}$, whereas in the 
forward rules $\overrightarrow{{Return1}}$ and 
$\overrightarrow{{Return2}}$, the trigger was reaching 
the $\ms{stop}$ statement.

\begin{figure}[tp]
\begin{mathpar}
\inferrule*[leftstyle=\rm,left=$\overleftarrow{\mathit{IfTrue1}}$] 
{ 
\block(\ell_1) = e_1\\
\\
\tuple{\sigma}{e_1} \Downarrow	 v_1 \\ 
\ms{is\_true?}(v_1) \\
\ell_1 \iedge \ell \\
} 
{ 
\config{\sigma,\ell_1,\ell'_1,\ms{if\_true}(\ell_1,\ell_2)\cons\s} 
    \lh
\config{\sigma,\ell,\ell_1,\s}
}
\and 
\inferrule*[leftstyle=\rm,left=$\overleftarrow{\mathit{IfTrue2}}$] 
{ 
\block(\ell_2) = e_2\\
\context(\ell_2) = (\ms{if}~ [e_1]^{\ell_1} ~\ms{then}~ s_1 
~\ms{else}~s_2~\ms{fi} ~[e_2]^{\ell_2}) \\
\tuple{\sigma}{e_2} \Downarrow	 v_2 \\ 
\ms{is\_true?}(v_2) \\
}
{
\config{\sigma,\ell_2,\ell'_2,\s}
\lh
\config{\sigma,\exit(s_1),\ell_2,\ms{if\_true}(\ell_1,\ell_2)\cons\s} 
}
\and 
\inferrule*[leftstyle=\rm,left=$\overleftarrow{\mathit{IfFalse1}}$] 
{
\block(\ell_1) = e_1\\
\\
\tuple{\sigma}{e_1} \Downarrow	 v_1 \\ 
\ms{is\_false?}(v_1) \\
\ell_1 \iedge \ell\\
} 
{ 
\config{\sigma,\ell_1,\ell'_1,\ms{if\_false}(\ell_1,\ell_2)\cons\s} 
    \lh
\config{\sigma,\ell,\ell_1,\s}
}
\and 
\inferrule*[leftstyle=\rm,left=$\overleftarrow{\mathit{IfFalse2}}$] 
{ 
\block(\ell_2) = e_2\\
\context(\ell_2) = (\ms{if}~ [e_1]^{\ell_1} ~\ms{then}~ s_1 
~\ms{else}~s_2~\ms{fi} ~[e_2]^{\ell_2}) \\
\tuple{\sigma}{e_2} \Downarrow	 v_2 \\ 
\ms{is\_false?}(v_2) \\
}
{
\config{\sigma,\ell_2,\ell'_2,\s}
\lh
\config{\sigma,\exit(s_2),\ell_2,\ms{if\_false}(\ell_1,\ell_2)\cons\s} 
}
\and
\inferrule*[leftstyle=\rm,left=$\overleftarrow{\mathit{LoopMain}}$] 
{ 
\block(\ell_1) = e_1\\
\tuple{\sigma}{e_1} \Downarrow	 v_1\\ 
\ms{is\_true?}(v_1) \\
\ell_1 \iedge \ell \\
\ell \neq \exit(s_2)
} 
{
\config{\sigma,\ell_1,\ell',\ms{loop1}(\ell_1,s_1,\ell_2,s_2)\cons\s}
\lh
\config{\sigma,\ell,\ell_1,\s}
}
\and
\inferrule*[leftstyle=\rm,left=$\overleftarrow{\mathit{LoopBase}}$] 
{ 
\block(\ell_2) = e_2 \\
\context(\ell_2) =
(\ms{from}~ [e_1]^{\ell_1} ~\ms{do}~ s_1 ~\ms{loop}~s_2~\ms{until} ~[e_2]^{\ell_2})\\
\tuple{\sigma}{e_2} \Downarrow	 v_2\\ 
\ms{is\_true?}(v_2) \\
} 
{
\config{\sigma,\ell_2,\ell'_2,\s}
\lh
\config{\sigma,\exit(s_1),\ell_2,\ms{loop1}(\ell_1,s_1,\ell_2,s_2)\cons\s}
}
\and
\inferrule*[leftstyle=\rm,left=$\overleftarrow{\mathit{Loop1}}$] 
{ 
\block(\ell_2) = e_2 \\
\tuple{\sigma}{e_2} \Downarrow	 v_2\\ 
\ms{is\_false?}(v_2) 
} 
{
\config{\sigma,\ell_2,\ell',\ms{loop2}(\ell_1,s_1,\ell_2,s_2)\cons\s}
\lh
\config{\sigma,\exit(s_1),\ell_2,\ms{loop1}(\ell_1,s_1,\ell_2,s_2)\cons\s}
}
\and
\inferrule*[leftstyle=\rm,left=$\overleftarrow{\mathit{Loop2}}$] 
{ 
\block(\ell_1) = e_1\\
\tuple{\sigma}{e_1} \Downarrow	 v_1\\ 
\ms{is\_false?}(v_1) 
} 
{
\config{\sigma,\ell_1,\ell',\ms{loop1}(\ell_1,s_1,\ell_2,s_2)\cons\s} 
\lh
\config{\sigma,\exit(s_2),\ell_1,\ms{loop2}(\ell_1,s_1,\ell_2,s_2)\cons\s}  
}
\end{mathpar}
\caption{Reversible backward semantics: rules for if and for loop} \label{fig:reversible-ifloop-backward}
\end{figure}

The backward rules for conditional and loop statements are
in Figure~\ref{fig:reversible-ifloop-backward}. 
As before, each rule $\overleftarrow{\mathit{Rule}}$ represents 
the inverse version of the corresponding forward rule 
$\overrightarrow{{Rule}}$ from Fig.~\ref{fig:reversible-ifloop}. 
Let us emphasize only that conditionals and loops are now \emph{undone} 
from end to start; consequently, the rules 
$\overleftarrow{\mathit{IfTrue2}}$ and 
$\overleftarrow{\mathit{IfFalse2}}$ now inspect the context 
of the label to determine whether it corresponds to a 
conditional assertion or a loop test, and push a new
element onto the stack (while this was done
in the forward semantics in rules 
$\overrightarrow{\mathit{IfTrue1}}$ and
$\overrightarrow{\mathit{IfFalse1}}$).

Similarly, for loops, it is rule $\overleftarrow{\mathit{LoopBase}}$ 
that now inspects the context (and push a new element 
onto the stack). 
Conversely, the loop now terminates with the 
rule $\overleftarrow{\mathit{LoopMain}}$. 
Here, we verify that $\ell \neq \exit(s_2)$, since there 
are two outgoing edges in the inverse CFG: one that 
backtracks to the statement preceding the loop 
(the one of interest), and another one that reverses 
a further iteration by jumping to the exit of $s_2$. 
This condition ensures that the correct path is followed.

\begin{conference}
\begin{example} \label{ex:backward_trace}
  Consider the derivation shown in Example~\ref{ex:forward_trace}.
  Starting from the last configuration, we have the 
  backward derivation shown in Figure~\ref{fig:backward_derivation}.
  As in the previous example, we highlight 
  in red the ``active'' label in each step (i.e., the label 
  that represents the last executed statement). 
  \begin{figure}[tbp]
  $
  \begin{array}{rl}
  & \config{\{n\mapsto 4,i\mapsto 2,total\mapsto 2\},\red{3},4,\nul}\\  
  
    {}_{\overleftarrow{\mathit{Return1}}} & 
  \config{\{n\mapsto 4,i\mapsto 2,total\mapsto 2\},\red{14},15,\ms{call}(3)\cons\nul}\\  

    {}_{\overleftarrow{\mathit{AssVar}}} & 
  \config{\{n\mapsto 2,i\mapsto 2,total\mapsto 2\},\red{13},14,\ms{call}(3)\cons\nul}\\  
  
  {}_{\overleftarrow{\mathit{LoopBase}}} & 
  \config{\{n\mapsto 2,i\mapsto 2,total\mapsto 2\},\red{11},13,\ms{loop1}(7,s_1,13,s_2)\cons\ms{call}(3)\cons\nul}\\  
  
    {}_{\overleftarrow{\mathit{IfTrue2}}} & 
  \config{\{n\mapsto 2,i\mapsto 2,total\mapsto 2\},\red{9},11,\ms{if\_true}(8,11)\cons\ms{loop1}(7,s_1,13,s_2)\cons\ms{call}(3)\cons\nul}\\  
  
  {}_{\overleftarrow{\mathit{AssVar}}} & 
  \config{\{n\mapsto 2,i\mapsto 2\},\red{8},9,\ms{if\_true}(8,11)\cons\ms{loop1}(7,s_1,13,s_2)\cons\ms{call}(3)\cons\nul}\\  
  
  {}_{\overleftarrow{\mathit{IfTrue1}}} &
    \config{\{n\mapsto 2,i\mapsto 2\},\red{7},8,\ms{loop1}(7,s_1,13,s_2)\cons\ms{call}(3)\cons\nul}\\  

%
%
%
%
%

  {}_{\overleftarrow{\mathit{Loop2}}} & 
    \config{\{n\mapsto 2,i\mapsto 2\},\red{12},7,\ms{loop2}(7,s_1,13,s_2)\cons\ms{call}(3)\cons\nul}\\  

  {}_{\overleftarrow{\mathit{AssVar}}} & 
    \config{\{n\mapsto 2,i\mapsto 1\},\red{13},12,\ms{loop2}(7,s_1,13,s_2)\cons\ms{call}(3)\cons\nul}\\  

  {}_{\overleftarrow{\mathit{Loop1}}} & 
    \config{\{n\mapsto 2,i\mapsto 1\},\red{11},13,\ms{loop1}(7,s_1,13,s_2)\cons\ms{call}(3)\cons\nul}\\  

  {}_{\overleftarrow{\mathit{IfFalse2}}} & 
    \config{\{n\mapsto 2,i\mapsto 1\},\red{10},11,\ms{if\_false}(8,11)\cons\ms{loop1}(7,s_1,13,s_2)\cons\ms{call}(3)\cons\nul}\\  

  {}_{\overleftarrow{\mathit{Skip}}} & 
    \config{\{n\mapsto 2,i\mapsto 1\},\red{8},10,\ms{if\_false}(8,11)\cons\ms{loop1}(7,s_1,13,s_2)\cons\ms{call}(3)\cons\nul}\\  

  {}_{\overleftarrow{\mathit{IfFalse1}}} & 
    \config{\{n\mapsto 2,i\mapsto 1\},\red{7},8,\ms{loop1}(7,s_1,13,s_2)\cons\ms{call}(3)\cons\nul}\\  

  {}_{\overleftarrow{\mathit{LoopMain}}} & 
    \config{\{n\mapsto 2,i\mapsto 1\},\red{6},7,\ms{call}(3)\cons\nul}\\
    
  {}_{\overleftarrow{\mathit{AssVar}}} & 
    \config{\{n\mapsto 2\},\red{5},6,\ms{call}(3)\cons\nul}\\

  {}_{\overleftarrow{\mathit{Call}}} & 
    \config{\{n\mapsto 2\},\red{2},3,\nul}\\
  
  {}_{\overleftarrow{\mathit{AssVar}}} & 
    \config{\emptystate,\red{1},2,\nul}\\
  \end{array}
  $
  \caption{Backward derivation with the reversible semantics}
  \label{fig:backward_derivation}
  \end{figure}
\end{example}
\end{conference}

\begin{journal}
\begin{example} \label{ex:backward_trace}
  Consider the derivation shown in Example~\ref{ex:forward_trace}.
  Starting from the last configuration, we have the 
  backward derivation shown in Figure~\ref{fig:backward_derivation}.
  As in the previous example, we highlight 
  in red the ``active'' label in each step (i.e., the label 
  that represents the last executed statement). 
  \begin{figure}[tbp]
  $
  \begin{array}{rl}
  & \config{\{n\mapsto 6,i\mapsto 3,total\mapsto 3\},\red{3},4,\nul}\\  
  
    {}_{\overleftarrow{\mathit{Return1}}} & 
  \config{\{n\mapsto 6,i\mapsto 3,total\mapsto 3\},\red{14},15,\ms{call}(3)\cons\nul}\\  

    {}_{\overleftarrow{\mathit{AssVar}}} & 
  \config{\{n\mapsto 3,i\mapsto 3,total\mapsto 3\},\red{13},14,\ms{call}(3)\cons\nul}\\  
  
  {}_{\overleftarrow{\mathit{LoopBase}}} & 
  \config{\{n\mapsto 3,i\mapsto 3,total\mapsto 3\},\red{11},13,\ms{loop1}(7,s_1,13,s_2)\cons\ms{call}(3)\cons\nul}\\  
  
    {}_{\overleftarrow{\mathit{IfTrue2}}} & 
  \config{\{n\mapsto 3,i\mapsto 3,total\mapsto 3\},\red{9},11,\ms{if\_true}(8,11)\cons\ms{loop1}(7,s_1,13,s_2)\cons\ms{call}(3)\cons\nul}\\  
  
  {}_{\overleftarrow{\mathit{AssVar}}} & 
  \config{\{n\mapsto 3,i\mapsto 3\},\red{8},9,\ms{if\_true}(8,11)\cons\ms{loop1}(7,s_1,13,s_2)\cons\ms{call}(3)\cons\nul}\\  
  
  {}_{\overleftarrow{\mathit{IfTrue1}}} &
    \config{\{n\mapsto 3,i\mapsto 3\},\red{7},8,\ms{loop1}(7,s_1,13,s_2)\cons\ms{call}(3)\cons\nul}\\  

  {}_{\overleftarrow{\mathit{Loop2}}} & 
    \config{\{n\mapsto 3,i\mapsto 3\},\red{12},7,\ms{loop2}(7,s_1,13,s_2)\cons\ms{call}(3)\cons\nul}\\  

  {}_{\overleftarrow{\mathit{AssVar}}} & 
    \config{\{n\mapsto 3,i\mapsto 2\},\red{13},12,\ms{loop2}(7,s_1,13,s_2)\cons\ms{call}(3)\cons\nul}\\  

   {}_{\overleftarrow{\mathit{Loop1}}} & 
    \config{\{n\mapsto 3,i\mapsto 2\},\red{11},13,\ms{loop1}(7,s_1,13,s_2)\cons\ms{call}(3)\cons\nul}\\  

   {}_{\overleftarrow{\mathit{IfFalse2}}} & 
    \config{\{n\mapsto 3,i\mapsto 2\},\red{10},11,\ms{if\_false}(8,11)\cons\ms{loop1}(7,s_1,13,s_2)\cons\ms{call}(3)\cons\nul}\\  

   {}_{\overleftarrow{\mathit{Skip}}} & 
    \config{\{n\mapsto 3,i\mapsto 2\},\red{8},10,\ms{if\_false}(8,11)\cons\ms{loop1}(7,s_1,13,s_2)\cons\ms{call}(3)\cons\nul}\\  

   {}_{\overleftarrow{\mathit{IfFalse1}}} & 
    \config{\{n\mapsto 3,i\mapsto 2\},\red{7},8,\ms{loop1}(7,s_1,13,s_2)\cons\ms{call}(3)\cons\nul}\\  

  {}_{\overleftarrow{\mathit{Loop2}}} & 
    \config{\{n\mapsto 3,i\mapsto 2\},\red{12},7,\ms{loop2}(7,s_1,13,s_2)\cons\ms{call}(3)\cons\nul}\\  

  {}_{\overleftarrow{\mathit{AssVar}}} & 
    \config{\{n\mapsto 3,i\mapsto 1\},\red{13},12,\ms{loop2}(7,s_1,13,s_2)\cons\ms{call}(3)\cons\nul}\\  

  {}_{\overleftarrow{\mathit{Loop1}}} & 
    \config{\{n\mapsto 3,i\mapsto 1\},\red{11},13,\ms{loop1}(7,s_1,13,s_2)\cons\ms{call}(3)\cons\nul}\\  

  {}_{\overleftarrow{\mathit{IfFalse2}}} & 
    \config{\{n\mapsto 3,i\mapsto 1\},\red{10},11,\ms{if\_false}(8,11)\cons\ms{loop1}(7,s_1,13,s_2)\cons\ms{call}(3)\cons\nul}\\  

  {}_{\overleftarrow{\mathit{Skip}}} & 
    \config{\{n\mapsto 3,i\mapsto 1\},\red{8},10,\ms{if\_false}(8,11)\cons\ms{loop1}(7,s_1,13,s_2)\cons\ms{call}(3)\cons\nul}\\  

  {}_{\overleftarrow{\mathit{IfFalse1}}} & 
    \config{\{n\mapsto 3,i\mapsto 1\},\red{7},8,\ms{loop1}(7,s_1,13,s_2)\cons\ms{call}(3)\cons\nul}\\  

  {}_{\overleftarrow{\mathit{LoopMain}}} & 
    \config{\{n\mapsto 3,i\mapsto 1\},\red{6},7,\ms{call}(3)\cons\nul}\\
    
  {}_{\overleftarrow{\mathit{AssVar}}} & 
    \config{\{n\mapsto 3\},\red{5},6,\ms{call}(3)\cons\nul}\\

  {}_{\overleftarrow{\mathit{Call}}} & 
    \config{\{n\mapsto 3\},\red{2},3,\nul}\\
  
  {}_{\overleftarrow{\mathit{AssVar}}} & 
    \config{\emptystate,\red{1},2,\nul}\\
  \end{array}
  $
  \caption{Backward derivation with the reversible semantics}
  \label{fig:backward_derivation}
  \end{figure}
\end{example}
\end{journal}

In order to prove the loop lemma below, stating that each forward step can be undone by a backward step and vice versa, we need to restrict our
attention to reachable configurations.

\begin{definition}[reachable configuration]
  A configuration $\config{\sigma,\ell'_1,\ell'_2,\s}$ is \emph{reachable} if it can be obtained by applying the rules of the semantics starting from an initial configuration, namely $\config{\emptystate,\ell_1,\ell_2,\nul} (\rh \cup \lh)^* \config{\sigma,\ell'_1,\ell'_2,\s}$ for some initial configuration $\config{\emptystate,\ell_1,\ell_2,\nul}$.
\end{definition}
This rules out, e.g., the case of configurations
$\config{\sigma,\ell,\ell',\s}$ where $\ell$ and $\ell'$ are 
not the two ends of an edge of the CFG.


\begin{lemma}[loop lemma]
  Let $\config{\sigma,\ell_1,\ell_2,\s}$ be a reachable configuration.
  Then, 
  $\config{\sigma,\ell_1,\ell_2,\s} \rh \config{\sigma',\ell'_1,\ell'_2,\s'}$
  iff
  $\config{\sigma',\ell'_1,\ell'_2,\s'} \lh \config{\sigma,\ell_1,\ell_2,\s}$.
\end{lemma}
\begin{journal}
\begin{proof}
  We prove the claim by a case distinction on the applied rule:
  \begin{description}
  \item[\rm ($\mathit{AssVar}$ and $\mathit{AssArr}$)] 
  In both cases, the forward rule performs the inverse function 
  of the backward rule, and vice versa. 
  In the case of $\mathit{AssVar}$, we observe that
  if $\sigma(x) = v_0$ and it becomes 
  $v_1 = \llbracket \oplus \rrbracket(\sigma(x),v)$ after the 
  forward step, where $\tuple{\sigma}{e}\Downarrow v$,
  yielding the store $\sigma' = \sigma[x\mapsto v_1]$, 
  it is 
  straightforward to see that updating variable 
  $x$ in $\sigma'$ with 
  $\mathcal{I}_{op}\llbracket \oplus \rrbracket(\sigma'(x),v)$,
  as the backward rule does, will restore the original 
  value $v_0$ of $x$ and store $\sigma$ 
  since $\tuple{\sigma'}{e}\Downarrow v$
  too (because Janus requires $x$ not to occur in $e$).
  Moreover, $\ell' \edge \ell''$ iff $\ell'' \iedge \ell'$
  since this case involves an assignment, so there exists 
  only a single outgoing edge in the CFG.\footnote{Note that
  the only statements whose associated nodes in the CFG 
  may have two incoming or two outgoing edges are the 
  conditions (tests or assertions) of a conditional 
  or a loop. All other nodes always have a single incoming
  and outgoing edge.} 
  The case of rule $\mathit{AssArr}$ is analogous.
  
  \item[\rm ($\mathit{Call}$ and $\mathit{Return1}$)]
  Consider the configuration 
  $\config{\sigma,\ell,\ell',\s}$. Then, to apply a 
  forward step with rule $\overrightarrow{\mathit{Call}}$  
  we must have $\block(\ell') = (\ms{call}~\mathit{id})$,
  with $\ell \edge \ell'$ in the CFG.
  Hence, after the forward step, the control becomes
  $\ell_s,\ell''$, with 
  $\block(\ell_s) = \ms{start}$ 
  (i.e., $\entry(\Gamma(\mathit{id}))$) 
  and $\ell_s \edge \ell''$ in the CFG. 
  Moreover, the rule also pushes an element of the form
  $\ms{call}(\ell')$ onto the stack.
  The backward step with rule
  $\overleftarrow{\mathit{Call}}$   
  performs exactly the inverse actions:
  checks that $\block(\ell_s) = \ms{start}$ and, then, 
  recovers the control where
  $\ell'$ points to the original $\ms{call}$ 
  and $\ell$ to its previous statement, also removing
  the top element of the stack.
  
  The case for rules $\overrightarrow{\mathit{Return1}}$ and
  $\overleftarrow{\mathit{Return1}}$ is similar. Now,
  the forward rule triggers when the next statement is
  $\ms{stop}$. In this case, removes the element
  $\ms{call}$ from the stack and recovers as control of 
  the configuration the label of the original $\ms{call}$
  and that of the next statement. Conversely, the backward 
  rule triggers when when the last statement points to a $\ms{call}$,
  and then pushes a new $\ms{call}$ element onto the stack and
  recovers the labels of the last two statements of the
  procedure called.
  
  \item[\rm ($\mathit{UnCall}$ and $\mathit{Return2}$)] The case
  for these rules is perfectly analogous to the previous one.
  
  \item[\rm ($\mathit{Skip}$)] This case is trivial.
  
  \item[\rm ($\mathit{IfTrue1}$ and $\mathit{IfTrue2}$)]
  As for $\mathit{IfTrue1}$, 
  the forward rule is triggered when the label of the 
  next statement to be executed, $\ell_1$, points to the test 
  of a conditional of the form $(\ms{if}~ [e_1]^{\ell_1} 
  ~\ms{then}~ s_1 ~\ms{else}~s_2~\ms{fi} ~[e_2]^{\ell_2})$. 
  Then, if the condition $e_1$ evaluates to $true$, it updates 
  the control to $\ell_1,\ell'_1$, where 
  $\ell'_1 = \entry(s_1)$. Additionally, it pushes an element 
  of the form $\ms{if\_true}(\ell_1,\ell_2)$ onto the stack. 
  Conversely, the backward rule considers a configuration 
  in which the label of the last executed statement, $\ell_1$, 
  is the test of a conditional (the same one appearing in the 
  $\ms{if\_true}(\ell_1,\ell_2)$ element at the top of the 
  stack); if it evaluates to $true$, the rule backtracks to 
  the previous statement and pops the $\ms{if\_true}$ 
  element from the stack, thereby restoring the original 
  configuration.
  
  Regarding $\mathit{IfTrue2}$, the forward rule is triggered 
  when the label of the next statement to be executed, 
  $\ell_2$, points to the assertion of a conditional 
  $(\ms{if}~ [e_1]^{\ell_1} ~\ms{then}~ s_1~\ms{else}~s_2~\ms{fi} ~[e_2]^{\ell_2})$ 
  and an element of the form $\ms{if\_true}(\ell_1,\ell_2)$ is 
  at the top of the stack. In this case, it pops the element 
  from the stack and updates the control to advance to the 
  next statement. Conversely, the backward rule considers 
  a configuration in which the label of the last executed 
  statement, $\ell_2$, is the assertion of a conditional 
  (verified via the auxiliary function $\context$); 
  it then backtracks to the previous statement, $\exit(s_1)$, 
  and pushes the element $\ms{if\_true}(\ell_1,\ell_2)$ 
  back onto the stack, thereby restoring the original 
  configuration.

  \item[\rm ($\mathit{IfFalse1}$ and $\mathit{IfFalse2}$)] The case
  for these rules is perfectly analogous to the previous one.

  \item[\rm ($\mathit{LoopMain}$ and $\mathit{LoopBase}$)] 
  Regarding $\mathit{LoopMain}$, the forward rule is triggered 
  when the label of the next statement to be executed, $\ell_1$, 
  points to the assertion $e_1$ of a loop of the form 
  $(\ms{from}~ [e_1]^{\ell_1} ~\ms{do}~ s_1 ~\ms{loop}~s_2~\ms{until} ~[e_2]^{\ell_2})$. 
  In this case, the derived configuration updates the label of the 
  next statement to point to $\entry(s_1)$ and pushes the 
  element $\ms{loop1}(\ell_1,s_1,\ell_2,s_2)$ onto the stack. 
  Conversely, the backward rule proceeds inversely when the 
  element at the top of the stack has the form 
  $\ms{loop1}(\ell_1,s_1,\ell_2,s_2)$, the label of the 
  last executed statement is $\ell_1$ (pointing to the loop 
  assertion $e_1$), and, furthermore, it evaluates to 
  $\mathit{true}$. In this case, the backward rule pops 
  the $\ms{loop1}$ element from the stack and updates the 
  control to include the label $\ell$ of the statement preceding 
  the loop, i.e., $\ell_1 \iedge \ell$ with $\ell\neq\exit(s_2)$.
  
  Regarding $\mathit{LoopBase}$, the situation is similar to 
  the previous case but now focus on the loop test 
  $e_2$. While the forward rule pops the element $\ms{loop1}$
  from the stack and updates the control to the statement 
  following the loop when $e_2$ evaluates to $\mathit{true}$, 
  the backward rule pushes the $\ms{loop1}$ element back
  and updates the control such that the last executed 
  statement is $\exit(s_1)$, the statement executed 
  immediately before the evaluation of $e_1$.
   
  \item[\rm ($\mathit{Loop1}$ and $\mathit{Loop2}$)] These
  cases are straightforward.
  \qed
  \end{description}
\end{proof}
\end{journal}


\section{Related Work, Conclusion and Future Work}\label{sec:related}

We proposed a small-step semantics for Janus based on a notion of
program counter and established its equivalence to a classical
small-step semantics based on a notion of stack. They are also
equivalent to Janus semantics from the
literature~\cite{YOKOYAMA201071,LoopSem,JanusSem,LLS24}. The Janus
literature also covers extensions of Janus with other constructs, such
as local variables~\cite{YOKOYAMA201071,LoopSem} and
stacks~\cite{YOKOYAMA201071,LoopSem}. As immediate next step, we will
extend our approach to cover such constructs as well.
We also plan to apply our approach to define small-step semantics to
other pure reversible languages in the literature. These include
object-oriented languages such as ROOPL~\cite{ROOPL1,ROOPL} and
Joule~\cite{Joule}, and functional languages such as RFun~\cite{Rfun}
and CoreFun~\cite{CoreFun}, which are all equipped with big-step
semantics.
									
The only small-step reversible semantics we are aware of are
for
domain-specific languages for specifying assembly sequences in industrial robots~\cite{revAss,revAss1}. However,
their semantics
is quite different from the one of Janus or functional or object-oriented reversible languages, since it involves the position and actions of the robot in the real world. Hence the relation with our work is quite limited.

As mentioned above, the literature has few semantics of high-level
languages based on a program counter. A notable exception is Gurevich
and Huggins' semantics of C \cite{GurevichH92} based on an
\emph{Abstract State Machine} (see also \cite{DomerGM02}). A similar
pattern is found in some intermediate representations like the
\emph{Static Single Assignment} (SSA), which rely on explicit
control-flow graphs---effectively reintroducing a program counter---to
facilitate data flow analysis \cite{CytronFRWZ91}

A promising area for future exploration is the integration of concurrency support into Janus. By enabling the creation of processes and enabling them to exchange messages, such an extension could allow one to write concurrent reversible programs. The definition of the reversible small-step semantics provided in this paper is a relevant first step in this direction.

\bibliographystyle{splncs04} 

\end{document}

%% file: mymacros.tex
\newcommand{\edge}{\longrightarrow}
\newcommand{\iedge}{\stackrel{\mbox{\tiny -1}}{\longrightarrow}}

\usepackage{amsmath}
\usepackage{amssymb}
\usepackage{hyperref}
\usepackage{mathtools}
\usepackage[all]{xy}
\usepackage{subfiles} 
\usepackage{mathpartir}

\usepackage{scalerel}
\usepackage{mathrsfs}
\usepackage{stmaryrd}
\usepackage{thm-restate}
\usepackage{quoting}
\usepackage{tikz}
\usetikzlibrary{matrix}

\newcommand{\red}[1]{\textcolor{red}{#1}}

\newcommand{\ms}[1]{\mathsf{#1}}

\newcommand \tuple[2]{ #1 \vdash #2}
\newcommand{\mi}[1]{\mathit{#1}}

\newcommand{\rules}[1]{\textsc{#1}}